\documentclass[journal]{IEEEtran}
\usepackage{amsmath,amsfonts,amssymb}  
\usepackage{amsthm}                    
\usepackage{amsbsy}                    
\usepackage{mathrsfs}                  
\newtheoremstyle{italicindent}
  {3pt}                
  {3pt}                
  {\itshape}           
  {\parindent}        
  {\itshape}           
  {:}                  
  {.5em}               
  {}                   
\theoremstyle{italicindent}
\newtheorem{theorem}{Theorem}
\newtheorem{proposition}{Proposition}
\makeatletter
\renewenvironment{proof}[1][\proofname]{\par
  \pushQED{\qed}%
  \normalfont \topsep6\p@\@plus6\p@\relax
  \trivlist
  \item[\indent\itshape #1:]\ignorespaces
}{%
  \popQED\endtrivlist\@endpefalse
}
\makeatother
\usepackage{algpseudocode}             
\usepackage{algorithm}                 
\usepackage{float}                     
\usepackage{array}                     
\usepackage{comment}
\usepackage{textcomp}                  
\usepackage{stfloats}                  
\usepackage{url}                       
\usepackage{verbatim}                  
\usepackage{subfigure}                 
\usepackage{subfig}                    
\usepackage{cite}                      
\usepackage{CJKutf8}                 
\usepackage{bm}
\usepackage{graphicx}  
\usepackage{booktabs}
\usepackage{multirow}
\usepackage{caption}
\begin{document}
\begin{CJK}{UTF8}{gbsn}
\title{Array Zooming Optimization for Near-Field Localization With Movable Antennas}
\author{Yuxin Duan, Boyu Teng, Xiaojun Yuan, \textit{Fellow, IEEE} and Rui Wang, \textit{Senior Member, IEEE}}
       
\maketitle

\begin{abstract}
The emergence of movable antenna (MA) technology provides a promising way to enhance wireless sensing and communication by introducing spatial degrees of freedom through dynamic array reconfiguration. 
In near-field localization, achieving high resolution at low cost necessitates the adoption of sparse arrays. However, such sparsity tends to introduce spatial ambiguity due to aliasing effects. 
To resolve this resolution-ambiguity dilemma, this paper proposes an MA-enabled array zooming (AZ) system. First, we design a multi-measurement array zooming system that dynamically adjusts antenna spacings. By fusing the observational information from different measurements, the proposed AZ system effectively mitigates spatial aliasing while maintaining spatial resolution. Second, to quantify the performance limits under the severe multi-modal distributions inherent in sparse near-field sensing, we theoretically analyze the false peak distribution and derive a tighter performance lower bound, which incorporates the false detection probability. Third, considering that multiple false peaks may exist in practical multi-modal distributions, we propose an optimization algorithm for the AZ system to suppress false peaks and minimize the localization error. Extensive numerical results demonstrate that the proposed AZ strategy adaptively optimizes array configurations under varying signal-to-noise ratios (SNRs), substantially outperforming both conventional fixed-spacing arrays and Cramér–Rao bound (CRB)-based AZ benchmarks in localization accuracy.
\end{abstract}

\begin{IEEEkeywords}
Movable antenna(MA) array, near-field localization, false peak suppression, antenna position optimization.
\end{IEEEkeywords}

\section{Introduction}
The evolution of sixth generation mobile communication systems toward higher frequency bands and massive antenna arrays has driven integrated sensing and communication (ISAC) to become a research focus \cite{ISAC1,ISAC2,ISAC3}. With the deployment of extremely large scale antenna arrays and the adoption of millimeter wave and terahertz bands, electromagnetic wave propagation shifts fundamentally from a far field planar wavefront model to a near field spherical wavefront model \cite{Lu2022ELAA, Cui2023NF, Han2024NFTutorial}. The near field effect enables precise distance estimation by leveraging distinct wavefront curvatures, significantly improving angular resolution and facilitating high precision single station localization. Consequently, near field localization provides transformative technical support for emerging scenarios like autonomous driving and industrial automation \cite{Liu2022ISAC, Zhang20196G}.

Achieving high spatial resolution, however, requires a large aperture array. In practice, realizing a large aperture with dense fixed antennas demands a large number of radio frequency chains and causes excessive hardware costs. Deploying sparse arrays provides an effective solution for reducing the number of antennas and radio frequency chains \cite{xishu_jiangben1,xishu_jiangben2}, but large inter-element spacing generated by sparse arrays may cause grating lobes, which leads to spatial ambiguity and degrade localization performance \cite{xishu_jiaodumohu1,xishu_jiaodumohu2}. In addition, traditional sparse arrays use fixed-position antennas, which cannot fully exploit spatial degrees of freedom and limit the ability to achieve high-resolution parameter estimation \cite{sparse1,sparse2}. 
These challenges motivate the need for spatially reconfigurable arrays that can compensate for sparsity across multiple measurements.

Movable antenna (MA) technology has emerged as a new paradigm to break the physical limitations of fixed-position antennas by dynamically adjusting the positions of antenna elements within a local continuous region \cite{Zhu_MA_Comms, Wong_FAS_Mag, B_ning, Tutorial}. Movable antenna technology has been shown to provide significant performance gains in wireless communication systems \cite{Performance,Ma_MIMO_Capacity,Multiuser_Communication2,Multiuser_Communication1,Beamforming2,Beamforming3,Beamforming,Communications}. Beyond far field scenarios, MA has been extended to near field communications \cite{Communications}, where antenna positions are optimized under spherical wave models to approach performance upper bounds for both digital and analog beamforming architectures.

Building upon the success in communication networks, existing studies extend movable antennas to wireless sensing and ISAC. In \cite{MA_Localization1}, one and two dimensional movable antenna coordinates are optimized to minimize the Cramér–Rao bound (CRB) of the estimation mean squared error, reducing angle estimation errors and mitigating angular ambiguity compared to static arrays. In ISAC systems, MAs are employed to jointly optimize communication and sensing performance, including maximizing the sum of communication rate and sensing mutual information or maximizing sensing  signal-to-interference-plus-noise ratio (SINR) under communication constraints, with algorithms developed for antenna position optimization \cite{ISAC4,ISAC5}. Moreover, MA‑aided ISAC can also leverage statistical channel state information (CSI) to design antenna positions, reducing movement overhead while satisfying sensing CRB constraints \cite{ISAC6}. 

To further exploit the reconfigurability of MA in sparse near‑field localization, this paper proposes an MA array zooming system. The array zooming system can dynamically adjust antenna spacings and fuse information from multiple different measurements, thereby mitigating spatial aliasing while preserving spatial resolution and achieving high-precision near-field localization. 
To achieve high parametric resolution in near-field localization while reducing hardware expenditure, a limited number of MAs must be distributed over a massive aperture array. This deployment often results in a highly sparse array geometry and generates a large number of local maxima of the log-likelihood function. 
Under such a multi-modal log-likelihood function, the CRB fails to serve as a tight performance lower bound.
To address this issue, we analyze the distribution characteristics of near-field false peaks for uniformly spaced sparse arrays and derive a tight performance lower bound to quantify the fundamental limit of near-field localization. 
Furthermore, we propose an optimization algorithm for the array zooming system. The objective function of this algorithm incorporates the false detection probabilities of multiple high false peaks and their corresponding estimation mean square error (MSE). 
The main contributions of this paper are summarized as follows:
\begin{itemize}
    \item We propose a MA enabled array zooming system for near-field localization based on maximum likelihood (ML) estimation. By exploiting the dynamic reconfigurability of MAs, the array zooming system enables multiple measurements and fuses the information from different measurements to achieve high-precision near-field localization. 
    We characterize the false peak distribution of uniformly spaced MA sparse arrays in the near field, providing a theoretical foundation for subsequent array zooming optimization.
    
    \item We derive a new MSE lower bound to quantify the performance limit of near-field localization under multi-modal log-likelihood distributions. Unlike the CRB, the derived MSE lower bound accounts for the false detection probability that a false peak exceed the true peak, thus serving as a tighter lower bound than the CRB in noisy environments.
    
    \item Furthermore, we develop an optimization algorithm for the array zooming system, which can adaptively determine the optimal antenna spacings under different signal-to-noise ratios (SNRs). 
    By formulating an objective function  that accounts for potential detection errors at multiple high false peaks, the proposed strategy effectively suppresses spatial aliasing while maintaining high parametric resolution. 
    Extensive numerical results demonstrate that the optimized array zooming system outperforms fixed-spacing arrays and CRB-based baselines, effectively suppressing false
    peaks and approaching the derived tight performance lower bound when the SNR $\ge0$ dB.
\end{itemize}

The remainder of this paper is organized as follows. Section II establishes the system model for 2D movable antenna array. Section III analyzes near-field false peak distribution and derives a tight MSE lower bound incorporating false detection probability. Section IV proposes an optimization algorithm for array zooming system. Numerical results and discussions are provided in Section V. Section VI concludes the whole work.

Notation: The bold lowercase letters (e.g., $\mathbf{x}$) and bold uppercase letters (e.g., $\mathbf{X}$) denote vectors and matrices, respectively. The $\ell_2$ norm of a vector $\mathbf{x}$ is denoted by $\|\mathbf{x}\|_2$. The $(\cdot)^{\mathrm{T}}$, $(\cdot)^*$, and $(\cdot)^\mathrm{H}$ denote the transpose, conjugate, and Hermitian, respectively. 
The overline $\overline{(\cdot)}$ denotes the complement of a set.
The $\mathcal{CN}(0, \sigma^2)$ denotes the circularly symmetric complex Gaussian distribution with zero mean and variance $\sigma^2$. 
The $\|\cdot\|_{F}$ denotes the Frobenius norm and the $\langle\cdot,\cdot\rangle$ denotes the inner product of two matrices.
The symbol $\mathcal{O}(\cdot)$ denotes the standard big-$O$ notation, characterizing the asymptotic computational complexity.

\section{Array Zooming System Model}

\subsection{System Description}
We consider an uplink single-user SIMO system as shown in Fig.~\ref{fig:System model}. The system consists of a BS equipped with a 2D MA array and a user with a single antenna. The MA array comprises $N_{\mathrm{B}} = N_{\mathrm{B}, x} \times N_{\mathrm{B}, y}$ MAs arranged in a uniform planar array (UPA) configuration, where $N_{\mathrm{B}, x}$ and $N_{\mathrm{B}, y}$ denote the number of MAs along the horizontal and vertical directions, respectively. In practice, UPAs are widely adopted for near-field localization systems due to their well-established design principles and implementation simplicity. To leverage the mobility of MAs while maintaining this structural advantage, we constrain all MAs to maintain uniform spacing $d^{(t)}$ at each measurement instant $t$. Specifically, the antenna spacing $d^{(t)}$ takes values from a discrete set $\mathcal{D} = \{d_\mathrm{min}, d_\mathrm{min} + \Delta d, \ldots, d_\mathrm{max}\}$, where $\Delta d$ denotes the adjustment step size. Based on the maximum spacing $d_\mathrm{max}$, the physical boundaries of the MAs are defined by the continuous 2D region $\mathcal{C} = [-\frac{A}{2}, \frac{A}{2}] \times [-\frac{B}{2}, \frac{B}{2}]$, where $A = (N_{\mathrm{B},x}-1)d_\mathrm{max}$ and $B = (N_{\mathrm{B},y}-1)d_\mathrm{max}$ represent the maximum array sizes in the $x$- and $y$-directions, respectively. This configuration allows the array to operate in different modes, as illustrated in Fig.~\ref{fig:System model}. When $d^{(t)}$ increases, the array operates in  a ``zoom-out'' mode. Conversely, when $d^{(t)}$ decreases, the array operates in a “zoom-in” mode. This dynamic reconfiguration capability, referred to as array zooming, enables the BS to optimize localization performance adaptively. Accordingly, the 2D position of the $(i, j)$-th MA in the array plane at time instant $t$ is given by $\mathbf{s}_{i,j}^{(t)} = [i d^{(t)}, j d^{(t)}]^{\mathrm{T}}$, where $i \in \mathcal{I}_{N_{\mathrm{B}, x}}$ and $j \in \mathcal{I}_{N_{\mathrm{B}, y}}$ are indices from the index set $\mathcal{I}_{N} = \{-\lceil\frac{N-1}{2}\rceil, \ldots, \lfloor\frac{N-1}{2}\rfloor\}$.

\begin{figure}
    \centering
    \captionsetup{labelsep=period, font=small}
    \includegraphics[width=0.9\linewidth]{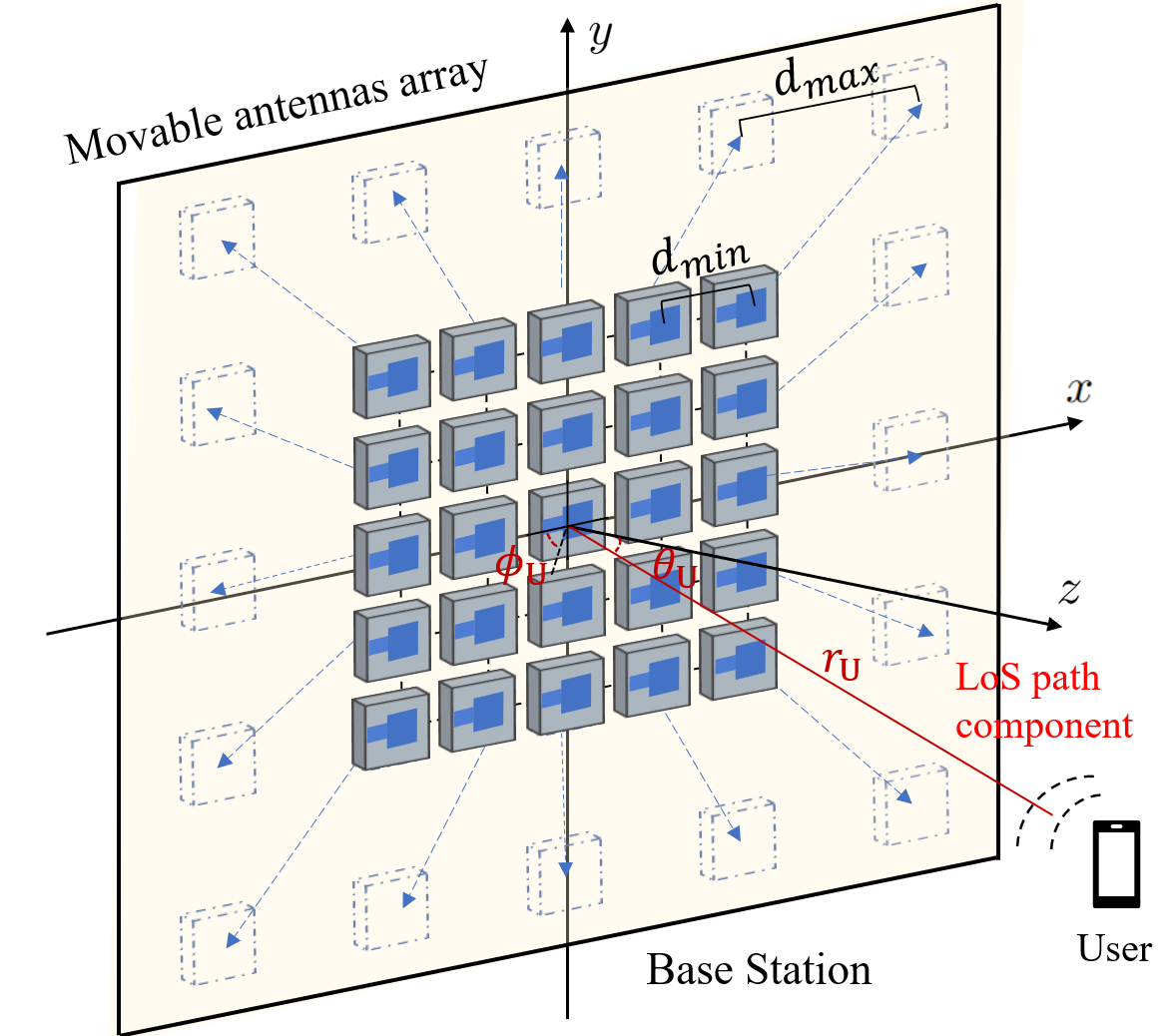}
    \caption{System model. The array operates in a “zoom-out” mode with antenna spacing increasing.}
    \label{fig:System model}
\end{figure}

\begin{figure*}[t!]
\centering
\captionsetup{labelsep=period, font=small}
\subfigure[]
{
\begin{minipage}[b]{0.45\textwidth}
\centering
\includegraphics[width=\textwidth]{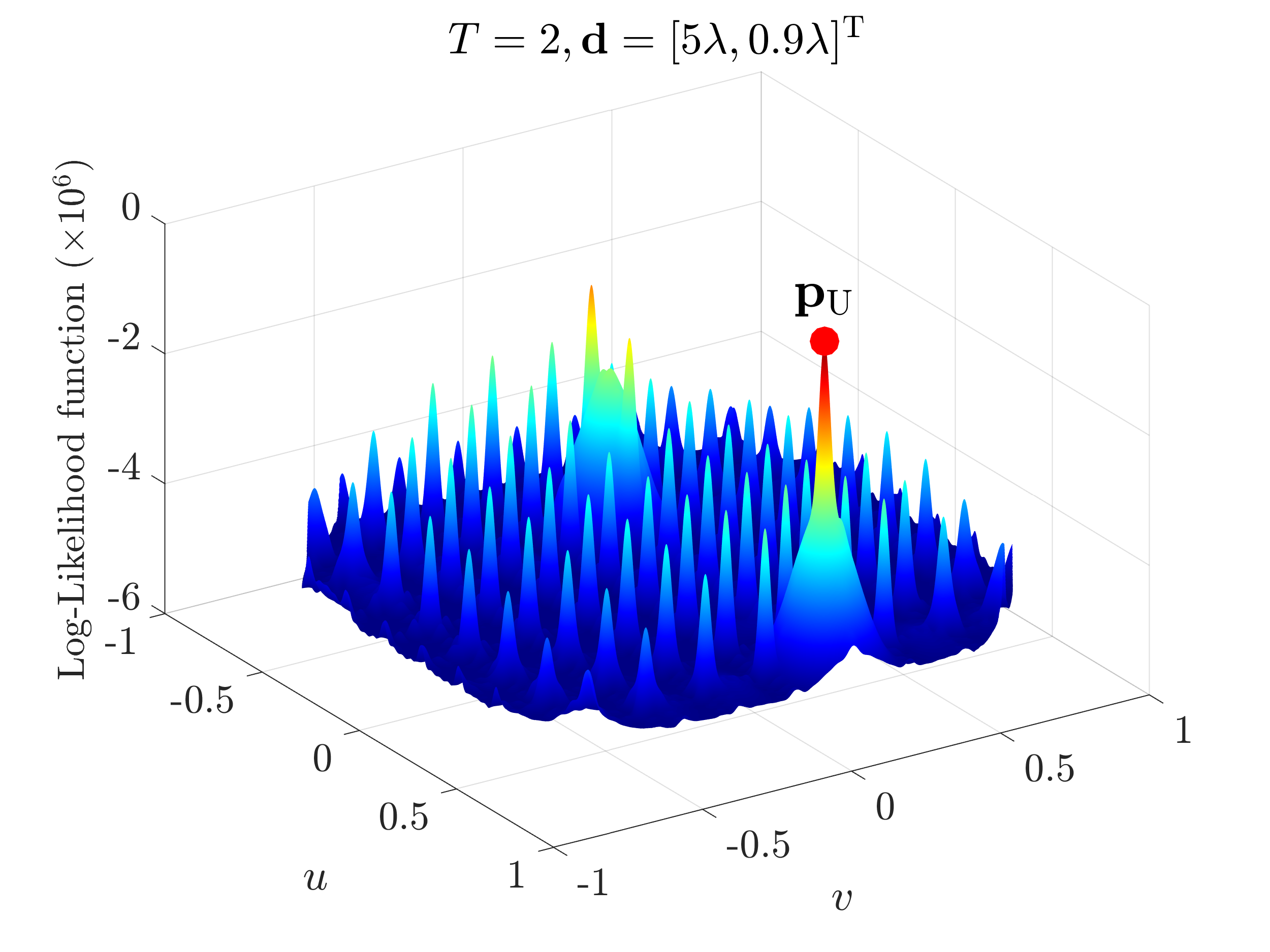}
\end{minipage}
}
\subfigure[]
{
\begin{minipage}[b]{0.45\textwidth}
\centering
\includegraphics[width=\textwidth]{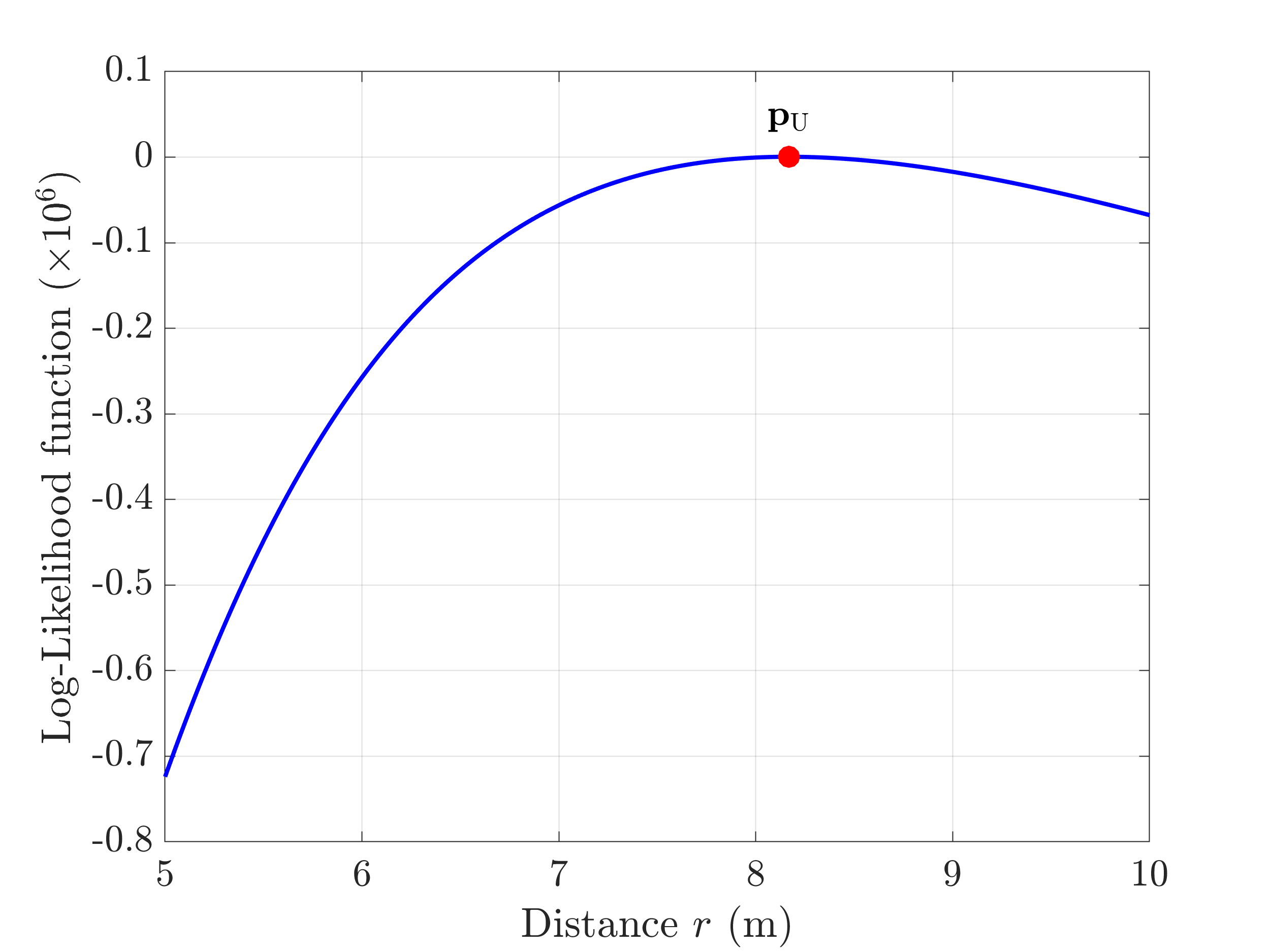}
\end{minipage}
}
\caption{Illustrations of a single-user near-field localization scenario where the user is located at distance $r=r_{\mathrm{U}}=8.168$ m from the BS array center. The SNR is fixed at 50 dB. The red point denotes the user’s true position. (a) The distribution in \eqref{eq:loglikelihood_function} with $T=2$, $\mathbf{d}=[5\lambda,0.9\lambda]$, $r=r_{\mathrm{U}}$ and $u,v$ are the cosines of the AoAs along the $x$-axis and $y$-axis.  (b) The distribution in \eqref{eq:loglikelihood_function} with $T=2$, $\mathbf{d}=[5\lambda,0.9\lambda]$ and $(u,v) = (u_{\mathrm{U}},v_{\mathrm{U}})=(0.717, 0.093)$. }
\label{fig:likelihood_comparison_25dB}
\end{figure*}

We establish a 3D Cartesian coordinate system centered at the BS, where axes $x$ and $y$ are defined as the horizontal and vertical directions of the 2D MA array plane, respectively, while axis $z$ is perpendicular to the array plane. Denote by $\mathbf{p}_{\mathrm{U}} = [x,y,z]^{\mathrm{T}}$ the user position and by $\mathbf{p}_{\mathrm{BS}}= \mathbf{0}$ the center of the BS array. Accordingly, $r_{\mathrm{U}} = \|\mathbf{p}_{\mathrm{BS}}-\mathbf{p}_{\mathrm{U}}\|$ denotes the distance between the user and the center of the BS array. The user position can be expressed by $\mathbf{p}_{\mathrm{U}} = [r_{\mathrm{U}}\sin\theta_{\mathrm{U}} \cos\phi_{\mathrm{U}},r_{\mathrm{U}}\sin\theta_{\mathrm{U}} \sin\phi_{\mathrm{U}},r_{\mathrm{U}}\cos\theta_{\mathrm{U}}]^{\mathrm{T}}$, where $\phi_{\mathrm{U}}$ and $\theta_{\mathrm{U}}$ are the azimuth and elevation angles of the user, respectively. Extending the 2D array coordinates $\mathbf{s}_{i,j}^{(t)}$ to this 3D space, the position of the $(i, j)$-th MA at time instant $t$ is expressed as $\mathbf{p}_{i, j}^{(t)} = [(\mathbf{s}_{i,j}^{(t)})^{\mathrm{T}}, 0]^{\mathrm{T}} = [i d^{(t)}, j d^{(t)}, 0]^{\mathrm{T}}$.

\subsection{Signal Model}
The BS estimates the user's location based on multiple measurements of the received signals. We consider a quasistatic flat-fading channel model, where the user's position remains unchanged.
We assume that the user is located in the near-field region of the BS array, i.e., $R_{\text{N}} < r < R_{\text{F}}$, where $R_{\text{N}} \triangleq \sqrt[3]{\dfrac{D^4}{8\lambda}}$ \cite{Fraunhofer} and $R_{\text{F}} \triangleq \dfrac{2D^2}{\lambda}$ \cite{Nearfield} are the Fresnel distance and the Fraunhofer distance respectively, and $D \triangleq d_\mathrm{max}\left((N_{\mathrm{B},x}-1)^2 + (N_{\mathrm{B},y}-1)^2\right)^{\frac{1}{2}}$ represents the maximum array aperture size, achieved when the antennas are fully extended. Furthermore, we assume that the considered near-field channel only consists of line-of-sight (LoS) components. 
Then, the channel coefficient between the user position and the $(i,j)$-th MA at the $t$-th measurement is given by 
\begin{equation}
\begin{aligned}
h_{(i,j)}^{(t)} = \alpha^{(t)} \frac{1}{r_{(i,j)}^{(t)}} e^{-\mathrm{j}\frac{2\pi}{\lambda}r_{(i,j)}^{(t)} },
\end{aligned}
\end{equation}
where $\mathrm{j}=\sqrt{-1}$, $\alpha^{(t)}$ is the complex channel gain including the position-independent terms at the $t$-th measurement. $r_{(i,j)}^{(t)} = \left\lVert \mathbf{p}_{(i,j)}^{(t)} - \mathbf{p}_{{\mathrm{U}}} \right\rVert$ represents the distance between the user and the $(i,j)$-th MA of the BS array at the $t$-th measurement and can be expressed as 
\begin{equation}
\begin{aligned}
r_{(i,j)}^{(t)} = \sqrt{(r_{\mathrm{U}} u_{\mathrm{U}} - id^{(t)})^2 + (r_{\mathrm{U}} v_{\mathrm{U}} - jd^{(t)})^2 + (r_{\mathrm{U}} \cos\theta_{\mathrm{U}})^2},
\end{aligned}
\end{equation}
where 
\begin{subequations}
\begin{align}
&u_{\mathrm{U}} \triangleq \sin\theta_{\mathrm{U}} \cos\phi_{\mathrm{U}},\\ &v_{\mathrm{U}} \triangleq \sin\theta_{\mathrm{U}} \sin\phi_{\mathrm{U}}.
\end{align}
\end{subequations}
Note that $u_{\mathrm{U}}$ (or $v_{\mathrm{U}}$) is the cosine of the angle between the user direction and the $x$-axis (or $y$-axis). Applying Taylor series expansion up to the second order, we obtain
\begin{equation}
\begin{aligned}
r_{(i,j)}^{(t)} &\approx r_{\mathrm{U}} - id^{(t)}u_{\mathrm{U}} - jd^{(t)}v_{\mathrm{U}} + \frac{i^2(d^{(t)})^2}{2r_{\mathrm{U}}}(1 - u_{\mathrm{U}}^2) \\
&+ \frac{j^2(d^{(t)})^2}{2r_{\mathrm{U}}}(1 - v_{\mathrm{U}}^2) - \frac{u_{\mathrm{U}} v_{\mathrm{U}}}{r_{\mathrm{U}}}ij(d^{(t)})^2.
\label{eq:dis}
\end{aligned}
\end{equation}

Following \cite{signal_amplitude}, 
we assume equal path loss from all MAs in the radiative near field. 
Then, the channel coefficient between the user position and the $(i,j)$-th MA at the $t$-th measurement can be expressed by
\begin{equation}
\begin{aligned}
h_{(i,j)}^{(t)} = \frac{\alpha^{(t)} }{r_{(0,0)}^{(t)}} e^{-\mathrm{j}\frac{2\pi}{\lambda}r_{(i,j)}^{(t)} }.
\end{aligned}
\end{equation}

By stacking all the $(i,j)$-th MA channel elements, we obtain the channel between the user and the BS at the $t$-th measurement as
\begin{equation}
\begin{aligned}
\bm{{\mathrm{H}}}^{(t)} =
\begin{bmatrix}
h_{(1,1)}^{(t)}& \cdots & h_{(1,N_y)}^{(t)}\\
\vdots & \ddots & \vdots \\
h_{(N_x,1)}^{(t)}& \cdots & h_{(N_x,N_y)}^{(t)}\\
\end{bmatrix}.
\end{aligned}
\end{equation}
Specifically, $\mathbf{H}^{(t)}$ can be expressed as $\mathbf{H}^{(t)} = \frac{\alpha^{(t)}}{r_{(0,0)}^{(t)}} \mathbf{A}(\mathbf{p}_{\mathrm{U}}, d^{(t)})$, where $\mathbf{A}(\mathbf{p}_{\mathrm{U}}, d^{(t)}) \in \mathbb{C}^{N_{\mathrm{B},x} \times N_{\mathrm{B},y}}$ denotes the near-field array response matrix. The $(i,j)$-th entry of $\mathbf{A}(\mathbf{p}_{\mathrm{U}}, d^{(t)})$ is given by
\begin{equation}
\left[\mathbf{A}(\mathbf{p}_{\mathrm{U}}, d^{(t)})\right]_{i,j} = e^{-\mathrm{j}\frac{2\pi}{\lambda} r_{i,j}^{(t)}}.
\label{eq:steering_t}
\end{equation}
Then, the received signal matrix at the $t$-th measurement is given by
\begin{equation}
\mathbf{Y}^{(t)} = \beta^{(t)} \mathbf{A}(\mathbf{p}_{\mathrm{U}}, d^{(t)}) + \mathbf{N}^{(t)},
\label{eq:signal_t}
\end{equation}
where $\beta^{(t)} = \frac{\alpha^{(t)}}{r_{(0,0)}^{(t)}}x^{(t)}$ with $x^{(t)}$ being the transmitted signal from the user at the $t$-th measurement and $\mathbf{N}^{(t)} \in \mathbb{C}^{N_{\mathrm{B},x} \times N_{\mathrm{B},y}}$ is the additive white Gaussian noise matrix at the $t$-th measurement with entries independently distributed as $\mathcal{CN}(0, \sigma^2)$ with $\sigma^2$ being the noise power.

\subsection{Location Estimation}
The BS estimates the user's position from the received signals by using the ML principle. For an arbitrary position $\mathbf{p} = [ru, rv, r\cos\theta]^{\mathrm{T}}$,  where $u=\sin\theta\cos\phi$ and $v=\sin\theta\cos\phi$ are the cosines of the AoAs along the $x$-axis and $y$-axis and $\phi$ and $\theta$ represent the azimuth and elevation angles of the arbitrary position respectively, the least-squares estimate of the complex amplitude $\beta^{(t)}$ is given by
\begin{equation}
\hat{\beta}^{(t)} = \frac{\langle\mathbf{A}(\mathbf{p},d^{(t)}),\mathbf{Y}^{(t)}\rangle}{\|\mathbf{A}(\mathbf{p},d^{(t)})\|_{F}^{2}},
\end{equation}
Substituting $\hat{\beta}^{(t)}$ into the log-likelihood function for a single measurement $t$ yields
\begin{equation}
\mathcal{L}(\mathbf{p}; d^{(t)}) = \frac{1}{\sigma^2}\frac{\left|\langle\mathbf{A}(\mathbf{p},d^{(t)}),\mathbf{Y}^{(t)}\rangle\right|^2}{\|\mathbf{A}(\mathbf{p},d^{(t)})\|_{F}^2} + c_t,
\label{eq:likelihood_single_measurement}
\end{equation}
where $c_t = -N_{\mathrm{B}}\ln(\pi\sigma^2) - \frac{1}{\sigma^2}\|\mathbf{Y}^{(t)}\|_{F}^2$ is a constant independent of $\mathbf{p}$.
For $T$ independent measurements, the log-likelihood function is expressed as
\begin{equation}
\mathcal{L}(\mathbf{p};\mathbf{d}) = \sum_{t=1}^{T} \mathcal{L}(\mathbf{p}; d^{(t)}) = \sum_{t=1}^{T}\frac{1}{\sigma^{2}}\frac{\left|\langle\mathbf{A}(\mathbf{p},d^{(t)}),\mathbf{Y}^{(t)}\rangle\right|^{2}}{\|\mathbf{A}(\mathbf{p},d^{(t)})\|_{F}^{2}} + c,
\label{eq:loglikelihood_function}
\end{equation}
where $\mathbf{d} = [d^{(1)},d^{(2)},\ldots,d^{(T)}]^{\mathrm{T}}$ is the antenna spacing configuration vector and $c = \sum_{t=1}^T c_t$.
The ML estimation for any given $\mathbf{d}$ is obtained by solving
\begin{equation}
\max_{\mathbf{p}} \mathcal{L}(\mathbf{p};\mathbf{d}).
\label{eq:max_loglikelihood_function}
\end{equation}
The ML estimation problem for near-field localization presents optimization challenges. This complexity stems from the strong non-convexity of the log-likelihood function in \eqref{eq:loglikelihood_function}, as illustrated in Fig.~\ref{fig:likelihood_comparison_25dB}(a). 
Specifically, spatial aliasing occurs when the antenna spacing exceeds half a wavelength, which generates false peaks in the spatial domain. These peaks manifest as local maxima in \eqref{eq:max_loglikelihood_function}, which complicates global optimization. 
Notably, Fig.~\ref{fig:likelihood_comparison_25dB}(b) reveals that the log-likelihood function exhibits a unique peak at $r_{\mathrm{U}}$ along the range dimension. We thus fix $r=r_{\mathrm{U}}$ and focus subsequent analysis on the angular domain, which is the primary source of the non-convexity of $\mathcal{L}(\mathbf{p}; \mathbf{d})$.

\section{Performance Analysis}
The likelihood values of false peaks may exceed that of the true peak in noisy environments, thereby degrading localization performance. For performance analysis, we now analyze the characteristics of near-field false peaks to lay a theoretical foundation for system optimization.

\begin{figure*}[t!]
\centering
\captionsetup{labelsep=period, font=small}
\subfigure[]
{
\begin{minipage}[b]{0.45\textwidth}
\centering
\includegraphics[width=\textwidth]{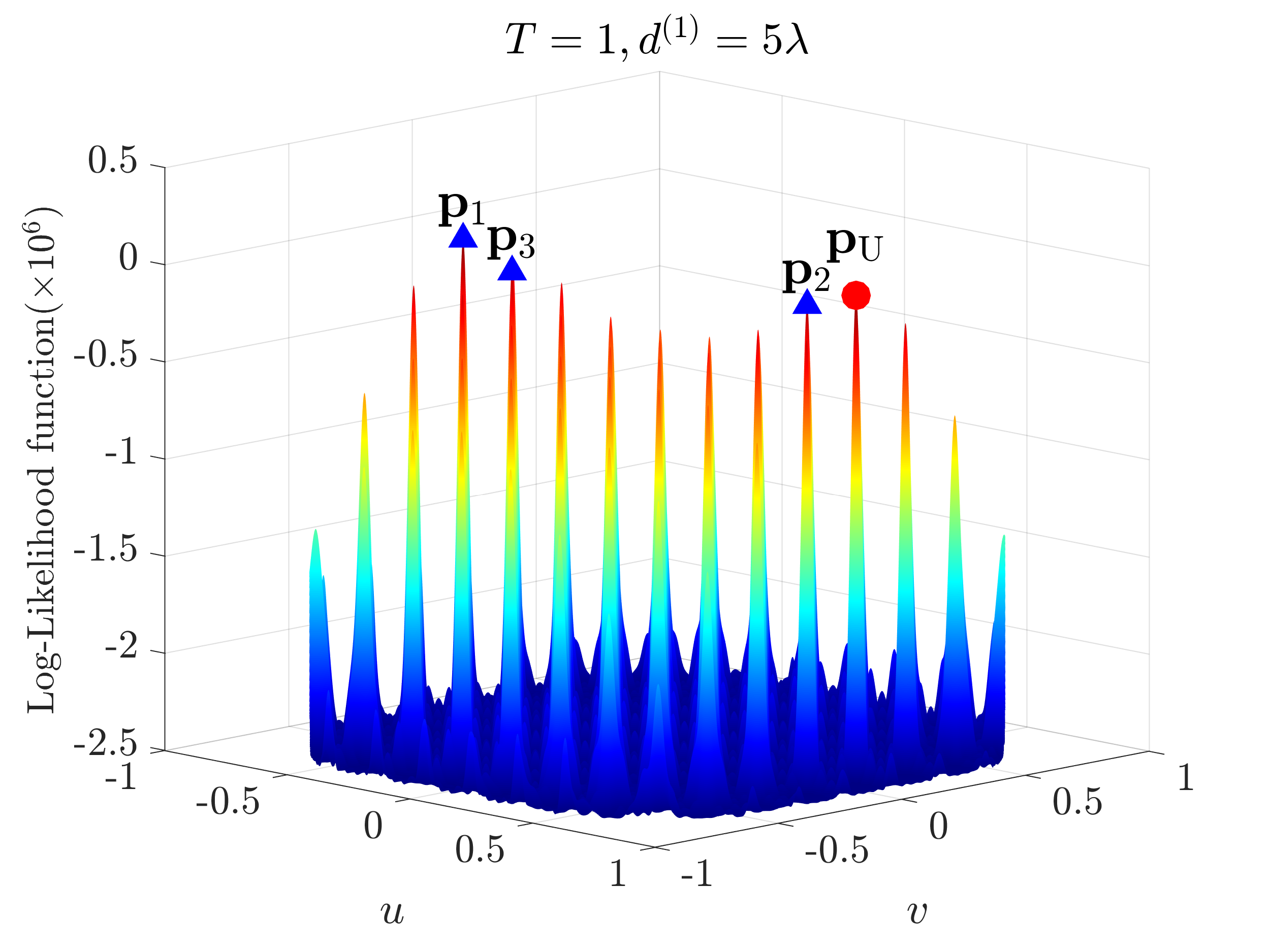}
\end{minipage}
}
\subfigure[]
{
\begin{minipage}[b]{0.45\textwidth}
\centering
\includegraphics[width=\textwidth]{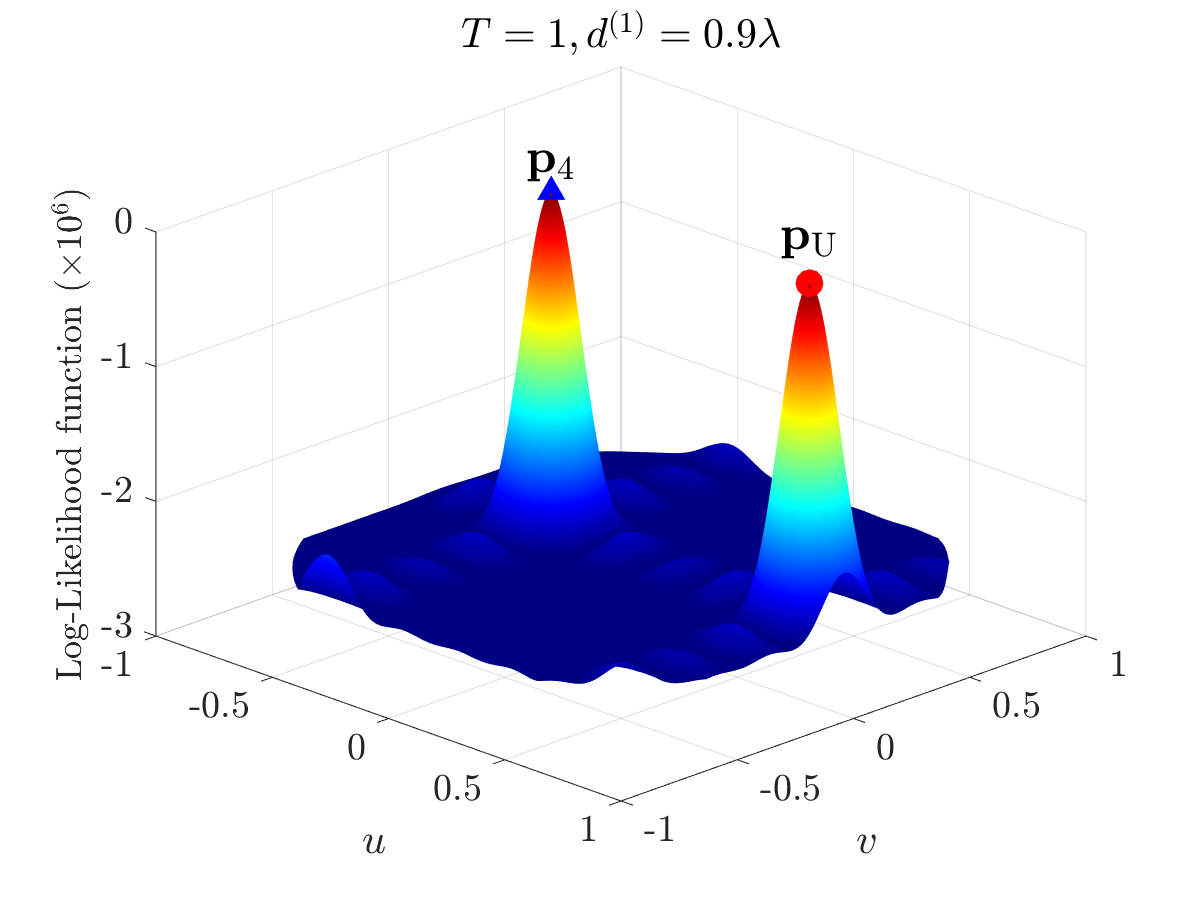}
\end{minipage}
}
\caption{Illustrations of distributions in \eqref{eq:loglikelihood_function} at SNR = 50 dB with $r=r_{\mathrm{U}}=8.168$ m, $T=1$ and $u,v$ are the cosine values of the angles along the $x$-axis and $y$-axis. The red point denotes the user’s position at $\mathbf{p}_{\mathrm{U}}=[5.856, 0.768, 5.642]^{\mathrm{T}}$ m with $\mathcal{L}(\mathbf{p}_{\mathrm{U}}; d^{(t)})=234.34$. The blue triangles denote the prominent false peaks $\mathbf{p}_i$.
(a) For the case of $d^{(1)}=5\lambda$, the peak coordinates are: $\mathbf{p}_1=[-5.603, -0.866, 5.880]^{\mathrm{T}}$ m with $\mathcal{L}(\mathbf{p}_1; d^{(1)})=-9.771\times10^3$; $\mathbf{p}_2=[4.223, 0.768, 6.949]^{\mathrm{T}}$ m with $\mathcal{L}(\mathbf{p}_2;d^{(1)})=-9.905\times10^4$; $\mathbf{p}_3=[-3.970, -0.866, 7.086]^{\mathrm{T}}$ m with $\mathcal{L}(\mathbf{p}_3;d^{(1)})=-1.275\times10^5$. 
(b) For the case of $d^{(1)}=0.9\lambda$, the peak coordinate is $\mathbf{p}_4=[-3.218, 0.768, 7.468]^{\mathrm{T}}$ m with $\mathcal{L}(\mathbf{p}_4;d^{(1)})=-87.23$.}
\label{fig:likelihood_comparison_25dB_single}
\end{figure*}

\subsection{Near-Field False Peak Analysis}
The peaks in a multi-measurement system stem from the combination of the false peaks of each individual measurement. As such, we begin with the analysis of the single-measurement case.

The positions of false peaks are mainly determined by the array geometry. Thus, we analyze them under the high-SNR assumption.
For a single measurement with antenna spacing $d^{(t)}$, we consider the ML estimation problem by maximizing \eqref{eq:likelihood_single_measurement}.
Since each element of $\mathbf{A}(\mathbf{p}, d^{(t)})$ has unit magnitude, the denominator is a constant. Substituting the signal model \eqref{eq:signal_t} at high SNR, maximizing \eqref{eq:likelihood_single_measurement} is equivalent to solving
\begin{align}
\max_{\mathbf{p}} \quad \left|\langle \mathbf{A}(\mathbf{p}, d^{(t)}), \mathbf{A}(\mathbf{p}_{\mathrm{U}}, d^{(t)}) \rangle\right|
\label{eq:grating_lobe_optimization}
\end{align}

We present the following proposition to characterize the solutions for problem \eqref{eq:grating_lobe_optimization}:
\begin{proposition}
For any single measurement $t$ with antenna spacing $d^{(t)}$, a location $\mathbf{p}$ is a solution to problem \eqref{eq:grating_lobe_optimization} if and only if the following conditions are satisfied:
\begin{subequations}
\label{eq:conditions}
\begin{align}
\frac{2d^{(t)}}{\lambda}(u_{\mathrm{U}} - u) &= k_1 , \quad k_1 \in \mathbb{Z} ,\label{eq:condition1} \\
\frac{2d^{(t)}}{\lambda}(v_{\mathrm{U}} - v) &= k_2, \quad k_2 \in \mathbb{Z} ,\label{eq:condition2} \\
\frac{(d^{(t)})^2}{\lambda}\left(\frac{1-u_{\mathrm{U}}^2}{r_{\mathrm{U}}} -\frac{1-u^2}{r}\right)  &= k_3 , \quad k_3 \in \mathbb{Z}, \label{eq:condition3} \\
\frac{(d^{(t)})^2}{\lambda}\left(\frac{1-v_{\mathrm{U}}^2}{r_{\mathrm{U}}} - \frac{1-v^2}{r}\right) &= k_4, \quad k_4 \in \mathbb{Z}, \label{eq:condition4} \\
\frac{(d^{(t)})^2}{\lambda}\left(\frac{u_{\mathrm{U}} v_{\mathrm{U}}}{r_{\mathrm{U}}} - \frac{u v}{r}\right) &= k_5, \quad k_5 \in \mathbb{Z}. \label{eq:condition5}
\end{align}
\end{subequations}
\end{proposition}

\begin{proof}
See Appendix~\ref{app:(14)}.
\end{proof}

We say that a false peak $\mathbf{p} \neq \mathbf{p}_{\mathrm{U}}$ of problem \eqref{eq:grating_lobe_optimization} is a grating lobe if $\mathbf{p}$ satisfies \eqref{eq:conditions}. 
Yet, no such a grating lobe exists in problem \eqref{eq:grating_lobe_optimization}, 
there is generally no exact solution of \eqref{eq:conditions} other than $\mathbf{p} = \mathbf{p}_{\mathrm{U}}$. 
However, the locations that approximately satisfy \eqref{eq:conditions} correspond to the primary false peaks of problem \eqref{eq:grating_lobe_optimization}.

To illustrate this, we examine the specific numerical results from Fig. \ref{fig:likelihood_comparison_25dB_single}, focusing on the top false peaks $\mathbf{p}_1, \dots, \mathbf{p}_4$. We substitute the parameters of each peak into the left-hand side of \eqref{eq:conditions} to obtain the continues value of each $k_i$, and use  $|k_i - \text{round}(k_i)|$ to quantify the deviation from the nearest integer.
For the highest false peak $\mathbf{p}_1$, the deviations are small across all constraints ($|k_i - \text{round}(k_i)| \leq 0.03$ for $i=1,\dots,5$)\footnote{There are some other factors contributing to the deviations. Specifically, the derivation of \eqref{eq:conditions} relies on the Fresnel approximation in \eqref{eq:dis}, whereas the numerical simulation employs the exact Euclidean distance model. Moreover, the presence of noise (e.g., at SNR = 50 dB in Fig. \ref{fig:likelihood_comparison_25dB_single}) introduces random perturbations to the peak positions compared to the noiseless theoretical model.}, indicating an approximate satisfaction of \eqref{eq:conditions}. A similar approximation of \eqref{eq:conditions} is observed for $\mathbf{p}_4$, where $|k_i - \text{round}(k_i)| < 0.005$ for $i=1,\dots,5$.
The peaks $\mathbf{p}_2$ and $\mathbf{p}_3$  exhibit relatively large deviations for \eqref{eq:condition3} with $|k_3 - \text{round}(k_3)| \approx 0.038$ and $0.043$, respectively.
This leads to their peak amplitudes being lower than that of $\mathbf{p}_1$.

\subsection{MSE Analysis}
In a multi-measurement system with $T > 1$, the local maxima of $\mathcal{L}(\mathbf{p}; \mathbf{d})$ in \eqref{eq:loglikelihood_function} are formed by the combination of false peaks from each single measurement. 
When ML estimator locates near the true peak, the MSE is lower-bounded by the CRB \cite{ Kay1993Estimation} as
\begin{equation}
\text{MSE}_{\mathrm{ML}}(\mathbf{p}_{\mathrm{U}}, \mathbf{d}) \geq \text{CRB}(\mathbf{p}_{\mathrm{U}}, \mathbf{d}).
\label{eq:crb_inequality}
\end{equation}
It is known that CRB is a loose lower bound especially when the likelihood in \eqref{eq:max_loglikelihood_function} is multi-modal.
To obtain a tighter lower bound, we consider the extension of \eqref{eq:grating_lobe_optimization} to the case of multiple measurements at high SNR, expressed as
\begin{equation}
     f(\mathbf{p};\mathbf{d}) = \sum_{t=1}^{T} \left|\left\langle \mathbf{A}(\mathbf{p}, d^{(t)}), \mathbf{A}(\mathbf{p}_{\mathrm{U}}, d^{(t)})\right\rangle\right|^2.
    \label{eq:multi_cor}
\end{equation}
Let $\mathcal{M}$ denote the set of all local maxima of \eqref{eq:multi_cor}, excluding the global maximum at $\mathbf{p}_{\mathrm{U}}$. The remaining local maxima are sorted in descending order of their function values of \eqref{eq:multi_cor}. Specifically, the primary false peak $\mathbf{p}_{\mathrm{F},1}$, the secondary false peak $\mathbf{p}_{\mathrm{F},2}$, and so on, are defined as
\begin{subequations}
\begin{align}
\mathbf{p}_{\mathrm{F},1} &= \arg\max_{\mathbf{p} \in \mathcal{M}}  f(\mathbf{p};\mathbf{d}),\\
\mathbf{p}_{\mathrm{F},k} &= \arg\max_{\mathbf{p} \in \mathcal{M} \setminus \{\mathbf{p}_{\mathrm{F},1}, \ldots, \mathbf{p}_{\mathrm{F},k-1}\}}  f(\mathbf{p};\mathbf{d}), \quad k = 2, 3, \ldots.
\end{align}
\label{eq:false_peaks}
\end{subequations}
We define an event of false detection as
\begin{align}
\mathcal{E}_{\mathrm{F},1} &\triangleq \left\{\mathcal{L}(\mathbf{p}_{\mathrm{U}}; \mathbf{d}) \leq \mathcal{L}(\mathbf{p}_{\mathrm{F},1}; \mathbf{d})\right\},
\end{align}
and the corresponding false detection probability is denoted by $ \Pr(\mathcal{E}_{\mathrm{F},1}\,|\,\mathbf{p}_{\mathrm{U}},\mathbf{d})$.\footnote{False detection involves multiple cases where the log-likelihood values of false peaks exceed that of the true peak. For simplicity, we first focus on the primary false peak and use ``false detection probability" to denote $ \Pr(\mathcal{E}_{\mathrm{F},1}\,|\,\mathbf{p}_{\mathrm{U}},\mathbf{d})$.}
By the law of total expectation, the MSE can be decomposed as
\begin{equation}
\begin{split}
\text{MSE}_{\text{ML}}(\mathbf{p}_{\mathrm{U}}, \mathbf{d}) &\triangleq \mathbb{E}\left[\|\hat{\mathbf{p}}_{\mathrm{ML}} - \mathbf{p}_{\mathrm{U}}\|^2 \mid  \mathbf{p}_{\mathrm{U}}, \mathbf{d}\right]\\
&= \text{MSE}_{1}(\mathbf{p}_{\mathrm{U}}, \mathbf{d}) \left(1 - \Pr(\mathcal{E}_{\mathrm{F},1}\mid\mathbf{p}_{\mathrm{U}},\mathbf{d})\right)  \\
&\quad +  \text{MSE}_{2}(\mathbf{p}_{\mathrm{U}}, \mathbf{d})  \Pr(\mathcal{E}_{\mathrm{F},1}\mid\mathbf{p}_{\mathrm{U}},\mathbf{d}),
\label{eq:mse_decomposition}
\end{split}
\end{equation}
where $\hat{\mathbf{p}}_{\mathrm{ML}}$ is the ML estimator of \eqref{eq:max_loglikelihood_function} and 
\begin{subequations}
\begin{align}
    \text{MSE}_{1}(\mathbf{p}_{\mathrm{U}}, \mathbf{d}) \triangleq \mathbb{E}\left[\|\hat{\mathbf{p}}_{\mathrm{ML}} - \mathbf{p}_{\mathrm{U}}\|^2 \mid \overline{\mathcal{E}_{\mathrm{F},1}}, \mathbf{p}_{\mathrm{U}}, \mathbf{d}\right], \\
    \text{MSE}_{2}(\mathbf{p}_{\mathrm{U}}, \mathbf{d}) \triangleq \mathbb{E}\left[\|\hat{\mathbf{p}}_{\mathrm{ML}} - \mathbf{p}_{\mathrm{U}}\|^2 \mid \mathcal{E}_{\mathrm{F},1}, \mathbf{p}_{\mathrm{U}}, \mathbf{d}\right]. \label{eq:MSE2_definition}
\end{align}
\label{eq:MSE_definitions}%
\end{subequations}
Here, $\text{MSE}_{1}(\mathbf{p}_{\mathrm{U}}, \mathbf{d})$ corresponds to the expected squared estimation error when the ML estimator occurs at the true peak, while $\text{MSE}_{2}(\mathbf{p}_{\mathrm{U}}, \mathbf{d})$ represents the expected squared estimation error when the ML estimator mistakenly selects the primary false peak $\mathbf{p}_{\mathrm{F},1}$.

Substituting \eqref{eq:crb_inequality} and \eqref{eq:MSE_definitions} into \eqref{eq:mse_decomposition}, the MSE of the ML estimator is lower-bounded by
the $\text{MSE}_\text{L}$ as
\begin{align}
   \text{MSE}_{\text{ML}}(\mathbf{p}_{\mathrm{U}}, \mathbf{d}) &\geq \text{MSE}_{\text{L}}(\mathbf{p}_{\mathrm{U}}, \mathbf{d}),
\label{eq:tight_rela} 
\end{align}
and 
\begin{equation}
\begin{split}
\text{MSE}_{\text{L}}(\mathbf{p}_{\mathrm{U}}, \mathbf{d})
   &= \left(1 - \Pr(\mathcal{E}_{\mathrm{F},1}\mid\mathbf{p}_{\mathrm{U}},\mathbf{d})\right)   \text{CRB}(\mathbf{p}_{\mathrm{U}}, \mathbf{d}) \\
   &+ \Pr(\mathcal{E}_{\mathrm{F},1}\mid\mathbf{p}_{\mathrm{U}},\mathbf{d})  \text{MSE}_{2}(\mathbf{p}_{\mathrm{U}}, \mathbf{d}),
\label{eq:weighted_mse} 
\end{split}
\end{equation}
where $\text{MSE}_{\text{L}}(\mathbf{p}_{\mathrm{U}}, \mathbf{d})$ is a tighter lower bound than CRB and the false detection probability $\Pr(\mathcal{E}_{\mathrm{F},1}\mid\mathbf{p}_{\mathrm{U}},\mathbf{d})$ is calculated by the following theorems. 
For notational simplicity, we let $\mathbf{A}_{\mathrm{U}}^{(t)} = \mathbf{A}(\mathbf{p}_{\mathrm{U}}, d^{(t)})$ and $\mathbf{A}_{\mathrm{F},1}^{(t)} = \mathbf{A}(\mathbf{p}_{\mathrm{F},1}, d^{(t)})$.

\begin{theorem}
The false detection probability $\Pr(\mathcal{E}_{\mathrm{F},1}\mid\mathbf{p}_{\mathrm{U}},\mathbf{d})$ is calculated as
\begin{equation}
\Pr(\mathcal{E}_{\mathrm{F},1}\mid\mathbf{p}_{\mathrm{U}},\mathbf{d}) = \frac{1}{2} + \frac{1}{\pi} \int_0^{\infty} \frac{\mathrm{Im}[\phi_{z}(u;\mathbf{p}_{\mathrm{U}},\mathbf{d})]}{u} du,
\label{eq:integral}
\end{equation}
where the characteristic function $\phi_{z}(u;\mathbf{p}_{\mathrm{U}},\mathbf{d})$ is given by
\begin{equation}
\phi_{z}(u;\mathbf{p}_{\mathrm{U}},\mathbf{d}) = \prod_{t=1}^{T} \frac{\exp\left(\mathrm{j}u (\boldsymbol{\mu}^{(t)})^{\mathrm{H}} \mathbf{J}(\mathbf{I} - \mathrm{j}u \boldsymbol{\Sigma}^{(t)} \mathbf{J})^{-1} \boldsymbol{\mu}^{(t)}\right)}{\det(\mathbf{I} - \mathrm{j}u \boldsymbol{\Sigma}^{(t)} \mathbf{J})}.
\end{equation}
Here, $\mathrm{j} = \sqrt{-1}$, $\mathbf{J} = \mathrm{diag}(1, -1)$, $z = \sum_{t=1}^{T} z^{(t)}$ and $z^{(t)}= \left|\frac{\langle\mathbf{A}_{\mathrm{F},1}^{(t)} ,\mathbf{Y}^{(t)}\rangle}{\sigma \|\mathbf{A}_{\mathrm{F},1}^{(t)}\|_F}\right|^2 - \left|\frac{\langle\mathbf{A}_{\mathrm{U}}^{(t)} ,\mathbf{Y}^{(t)}\rangle}{\sigma \|\mathbf{A}_{\mathrm{U}}^{(t)}\|_F}\right|^2$. $\boldsymbol{\mu}^{(t)} = \frac{\beta^{(t)}}{\sigma} \|\mathbf{A}_{\mathrm{U}}^{(t)}\|_F 
\begin{bmatrix} 
\rho^{(t)} \\ 1 
\end{bmatrix}$ and $\boldsymbol{\Sigma}^{(t)} =
\begin{bmatrix}
1 & \rho^{(t)} \\
(\rho^{(t)})^* & 1
\end{bmatrix}$ 
with $\rho^{(t)} = \frac{\langle\mathbf{A}_{\mathrm{F},1}^{(t)} ,\mathbf{A}_{\mathrm{U}}^{(t)}\rangle}{\|\mathbf{A}_{\mathrm{F},1}^{(t)}\|_F \|\mathbf{A}_{\mathrm{U}}^{(t)}\|_F}$.
\end{theorem}
\begin{proof}
See Appendix~\ref{app:pf_derivation}.
\end{proof}

\begin{figure}
    \centering
    \captionsetup{labelsep=period, font=small}
    \includegraphics[width=1\linewidth]{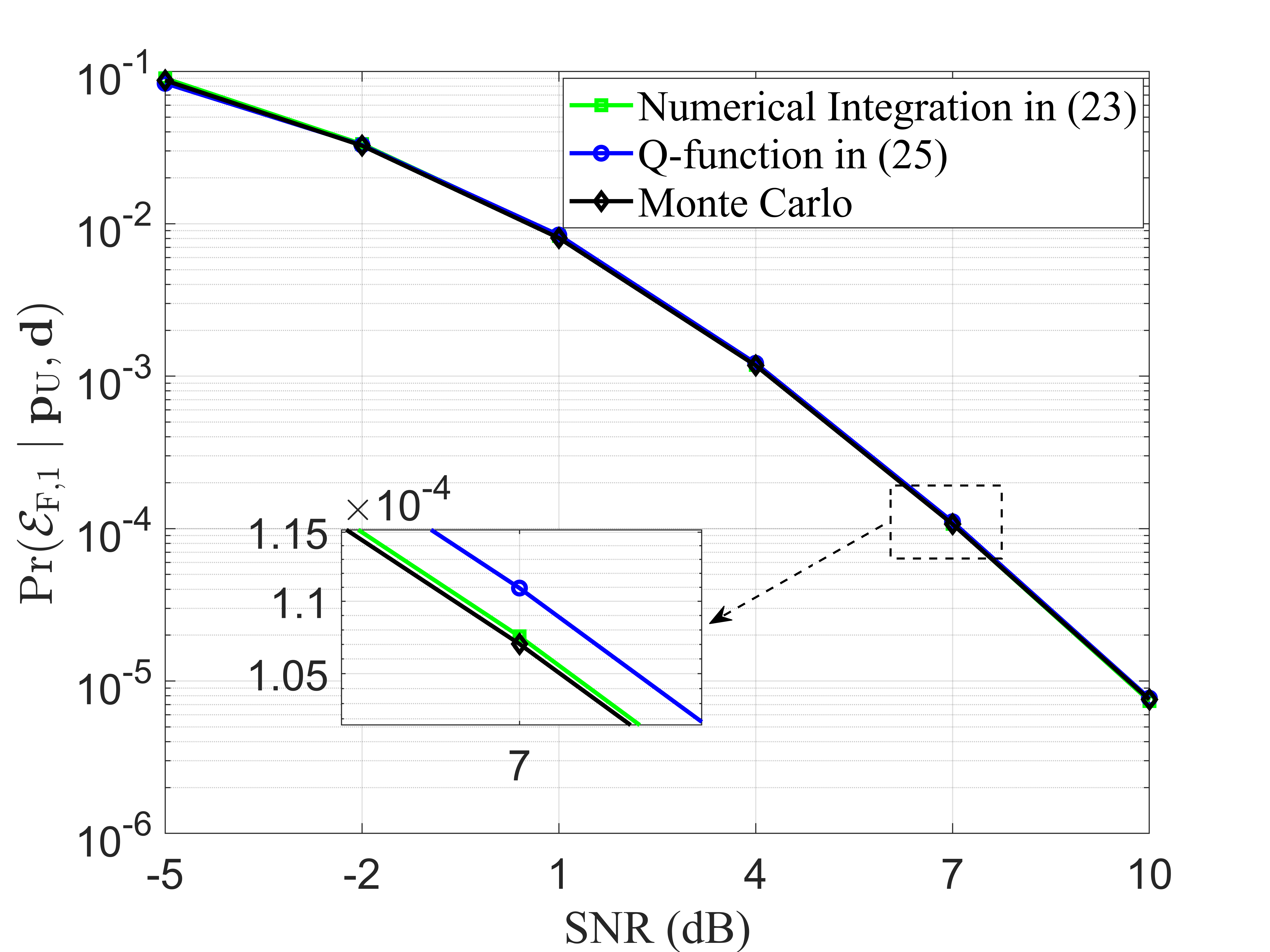}
    \caption{Comparison of $\Pr(\mathcal{E}_{\mathrm{F},1}\mid\mathbf{p}_{\mathrm{U}},\mathbf{d})$ obtained from the numerical integration in \eqref{eq:integral}, the Q-function in \eqref{eq:pf_qfunction_approximation} and the Monte Carlo method. The Monte Carlo result is estimated by 10,000 trials as the fraction of realizations in which the log-likelihood value at the false peak exceeds that at the true peak under complex Gaussian noise. All results are averaged over 50000 random user realizations and optimized spacing configurations.}
\end{figure}

\begin{theorem}
As $\sigma^2 \to 0$, the false detection probability $\Pr(\mathcal{E}_{\mathrm{F},1}\mid\mathbf{p}_{\mathrm{U}},\mathbf{d})$ satisfies the following limit:
\begin{equation}
\lim_{\sigma^2 \to 0} \frac{\Pr(\mathcal{E}_{\mathrm{F},1}\mid\mathbf{p}_{\mathrm{U}},\mathbf{d})}{Q\left(\sqrt{\frac{S(\mathbf{p}_{\mathrm{U}}, \mathbf{d})}{2\sigma^2}}\right)} = 1,
\label{eq:pf_qfunction_approximation}
\end{equation}
where the $Q$-function is the tail distribution function of the standard normal distribution \cite{Chiani} and 
\begin{align}
S(\mathbf{p}_{\mathrm{U}}, \mathbf{d})= \sum_{t=1}^T |\beta^{(t)}|^2 \left( \|\mathbf{A}_{\mathrm{U}}^{(t)}\|_F^2 - \frac{| \langle \mathbf{A}^{(t)}_{\mathrm{F},1}, \mathbf{A}^{(t)}_{\mathrm{U}} \rangle |^2}{\|\mathbf{A}_{\mathrm{F},1}^{(t)}\|_F^2} \right). \label{eq:signal_term_T}
\end{align}
\end{theorem}
\begin{proof}
See Appendix~\ref{app:pf_derivation_Q}.
\end{proof}

It is worth noting that the true complex amplitude $\beta^{(t)}$ is generally unknown in practice. For the computation of $\Pr(\mathcal{E}_{\mathrm{F},1}\mid\mathbf{p}_{\mathrm{U}},\mathbf{d})$ in Theorem 1 and Theorem 2, we substitute $\beta^{(t)}$ with its least-squares estimate $\hat{\beta}^{(t)}$.

Fig.~4 demonstrates that the Q-function approach achieves high accuracy in computing $\Pr(\mathcal{E}_{\mathrm{F},1}\,|\,\mathbf{p}_{\mathrm{U}},\mathbf{d})$, with negligible deviation from numerical integration and Monte Carlo benchmarks, even in the relatively low SNR regime. Crucially, the proposed Q-function method reduces the average execution time for $\Pr(\mathcal{E}_{\mathrm{F},1}\,|\,\mathbf{p}_{\mathrm{U}},\mathbf{d})$ computation by several orders of magnitude compared with numerical integration.  
Therefore, we adopt the Q-function in subsequent computations of 
$\Pr(\mathcal{E}_{\mathrm{F},1}\,|\,\mathbf{p}_{\mathrm{U}},\mathbf{d})$.

\section{Array Zooming Optimization}
Building upon the theoretical analysis of false peak distribution and the derivation of a tight MSE bound in Section III, we proceed to the design of the array zooming system. In specific, we formulate an optimization problem to determine the optimal antenna spacing configuration.

\subsection{Problem Formulation}
Since the exact location of the user is not known \textit{a priori}, we aim to minimize the maximum $\text{MSE}_\text{L}(\mathbf{p}, \mathbf{d})$ over the user region $\mathcal{R}$.  Based on \eqref{eq:weighted_mse}, the optimization problem is reformulated as
\begin{subequations}
\begin{align}
    &\min_{\mathbf{d}} \max_{\mathbf{p}\in \mathcal{R}} \text{MSE}_{\text{L}}(\mathbf{p}, \mathbf{d}) \\
    &\text{s.t.} \quad \mathbf{d} \in \mathcal{D}^T, 
\end{align}
\label{eq:P2}%
\end{subequations}
where $\mathcal{D}^T = \underbrace{\mathcal{D} \times \mathcal{D} \times \cdots \times \mathcal{D}}_{T \text{ times}}$ denotes the $T$-fold Cartesian product of $\mathcal{D}$. 
Directly solving problem \eqref{eq:P2} is computationally complex, primarily due to the lack of closed-form expressions for the local maxima of the multi-measurement log-likelihood function.
To address this challenge, we identify the primary false peak by searching in regions where high false peaks from different measurements are spatially consistent, meaning that their positions are sufficiently close in the angular domain such that they collectively contribute to a common false peak of the multi-measurement system.

\subsection{Primary False Peak Search}
We adopt a heuristic approach to search for the position of the primary false peak in the multi-measurement system.
We first exploit the near-field false peak analysis in Section III.A to locate the high false peak positions for each single measurement, and then identify regions where these false peaks are spatially consistent. A fine grid is generated in these spatially consistent regions to search for the position of the primary false peak $\mathbf{p}_{\mathrm{F},1}$. 

For a given position $\mathbf{p}\in \mathcal{R}$, the false peak can be identified using the conditions in \eqref{eq:conditions}.
Specifically, for each measurement $t \in \{1, \dots, T\}$, we first generate a candidate set of false peak positions $\mathcal{C}^{(t)} = \{(u^{(t)}_{k_{t,u}}, v^{(t)}_{k_{t,v}})\}_{k_{t,u},k_{t,v}\in\mathbb{Z}}$ based on:
\begin{subequations}
\begin{align}
u^{(t)}_{k_{t,u}} &= u_{\mathrm{U}} + k_{t,u}\frac{\lambda}{2d^{(t)}}, \quad k_{t,u}\in\mathbb{Z}, \label{eq:u_t_candidates}\\
v^{(t)}_{k_{t,v}} &= v_{\mathrm{U}} + k_{t,v}\frac{\lambda}{2d^{(t)}}, \quad k_{t,v}\in\mathbb{Z}, \label{eq:v_t_candidates}
\end{align}
\label{eq:uv_t_candidates}%
\end{subequations}
where each candidate corresponds to a position with $r = r_{\mathrm{U}}$.
We then filter the candidates in each $\mathcal{C}^{(t)}$ by using \eqref{eq:condition3}--\eqref{eq:condition5}. Specifically, for each candidate, we compute the coefficients:
\begin{subequations}
\begin{align}
k_3(u^{(t)}_{k_{t,u}}) &= \frac{d^{(t)2}}{\lambda} \left( \frac{u_{\mathrm{U}}^2-(u^{(t)}_{k_{t,u}})^2}{r_{\mathrm{U}}} \right), 
\label{eq:u_k3_candidates}\\
k_4(v^{(t)}_{k_{t,v}}) &= \frac{d^{(t)2}}{\lambda} \left( \frac{v_{\mathrm{U}}^2-(v^{(t)}_{k_{t,v}})^2}{r_{\mathrm{U}}} \right), 
\label{eq:u_k4_candidates}\\
k_5(u^{(t)}_{k_{t,u}},v^{(t)}_{k_{t,v}}) &= \frac{d^{(t)2}}{\lambda} \left( \frac{u_{\mathrm{U}}v_{\mathrm{U}} - u^{(t)}_{k_{t,u}}v^{(t)}_{k_{t,v}}}{r_{\mathrm{U}}} \right).
\label{eq:u_k5_candidates}
\end{align}
\label{eq:uv_k_candidates}%
\end{subequations}
A candidate is retained in the $\mathcal{C}^{(t)}$ if it satisfies the constraint:
\begin{align}
|k_i - \text{round}(k_i)| \le \epsilon, \quad  i \in \{3,4,5\},
\label{eq:deviation constraint}
\end{align}
where $k_i$ is a succinct notation for the coefficients on the left-hand side of  \eqref{eq:uv_k_candidates} and $\epsilon$ is a predefined tolerance that accounts for the deviation of the peaks from the satisfaction of \eqref{eq:conditions}. 

After executing the above procedure for all $T$ measurements, we obtain $T$ sets $\{\mathcal{C}^{(t)}\}_{t=1}^T$. 
Then, we identify regions where the false peaks from different sets are spatially consistent.
We define a false peak set as a collection of $T$ candidates: $\mathcal{H} = \{(u^{(t)}_{k_{t,u}}, v^{(t)}_{k_{t,v}})\}_{t=1}^T$, where $(u^{(t)}_{k_{t,u}}, v^{(t)}_{k_{t,v}}) \in \mathcal{C}^{(t)}$ for each $t=1,\ldots,T$. When $T \ge 2$, the $T$ candidates in $\mathcal{H}$ is considered spatially consistent only if every pair of candidates satisfies the spatial consistency condition. That is, for all $1 \le s < t \le T$, the spatial distance satisfies
\begin{equation}
\big|u^{(t)}_{k_{t,u}}-u^{(s)}_{k_{s,u}}\big|\le\varepsilon_u(t,s),\quad
\big|v^{(t)}_{k_{t,v}}-v^{(s)}_{k_{s,v}}\big|\le\varepsilon_v(t,s),
\label{eq:cons}
\end{equation}
where the axis-wise tolerances between measurements $t$ and $s$ are defined as
\begin{subequations}
\begin{align}
\varepsilon_u(t,s)&=\tfrac{1}{2}\!\left(\Delta u^{(t)}+\Delta u^{(s)}\right), \label{eq:axis_tol_u}\\
\varepsilon_v(t,s)&=\tfrac{1}{2}\!\left(\Delta v^{(t)}+\Delta v^{(s)}\right). \label{eq:axis_tol_v}
\end{align}
\label{eq:axis_tol}%
\end{subequations}
with 
\begin{equation}
\Delta u^{(t)}=\frac{\kappa\,\lambda}{N_x d^{(t)}} \quad \text{and} \quad
\Delta v^{(t)}=\frac{\kappa\,\lambda}{N_y d^{(t)}},
\label{eq:3db}
\end{equation}
being the main-lobe widths for the $t$-th measurement. $\kappa$ is a flexible beamwidth factor (e.g., $\kappa=0.891$ for 3dB beamwidth \cite{beamwidth}) that can be adjusted to balance search sensitivity and computational cost.

Suppose that there are $W$ false peak sets, denoted as $\mathcal{H}_1, \dots, \mathcal{H}_W$. For the $w$-th set $\mathcal{H}_w$, we generate a corresponding fine grid $\mathcal{G}_w$ with $N_w$  grid points. The geometric centroid of $\mathcal{H}_w$ is calculated by 
\begin{equation}
\bar{u}_w = \frac{1}{T}\sum_{t=1}^{T} u^{(t)}_{w}, \quad
\bar{v}_w = \frac{1}{T}\sum_{t=1}^{T} v^{(t)}_{w},
\label{eq:centroid}
\end{equation}
where $(u_{w}^{(t)}, v_{w}^{(t)})$ is the $t$-th peak in the $w$-th set $\mathcal{H}_w$. The local fine grid $\mathcal{G}_w$ is constructed as a rectangular grid centered at $(\bar{u}_w, \bar{v}_w)$. The coverage of the grid is defined by the half-widths $\delta_{u,w}$ and $\delta_{v,w}$:
\begin{equation}
\delta_{u,w} = \max_{1 \le t \le T}  \frac{\Delta u^{(t)}}{2},\quad
\delta_{v,w} = \max_{1 \le t \le T}  \frac{\Delta v^{(t)}}{2}.
\label{eq:radius}
\end{equation}
They are set to the maximum main-lobe half-widths among all $T$ configurations to ensure the peak is contained in the $\mathcal{G}_w$.
The step size of $\mathcal{G}_w$ along the $u$ and $v$ directions are defined as
\begin{equation}
\Delta u_w = \frac{\min_{t} \Delta u^{(t)}}{N_x}, \quad
\Delta v_w = \frac{\min_{t} \Delta v^{(t)}}{N_y},
\label{eq:grid_res}
\end{equation}
where $\min_{t} \Delta u^{(t)}$ and $\min_{t} \Delta v^{(t)}$ denote the finest beamwidths achieved during the measurement process, ensuring that the grid is finer than the sharpest main-lobe.

We define $\mathcal{G} = \bigcup\mathcal{G}_w$ as the set of all grid points. For each grid point $\mathbf{p}_{g} \in \mathcal{G}$, we compute the log-likelihood function and select the grid point with the highest log-likelihood value as
\begin{equation}
\mathbf{p}_{\mathrm{F},1,\mathrm{ini}} = \arg\max_{\mathbf{p}_g\in\bigcup\mathcal{G}} \mathcal{L}(\mathbf{p}_g;\mathbf{d}),
\label{eq:top_selection}
\end{equation}
where $\mathbf{p}_{{\mathrm{F},1,\mathrm{ini}}}$ serves as the initial point for gradient ascent (GA) method to search for the local maximum corresponding to the primary false peak position $\mathbf{p}_{{\mathrm{F},1}}$.
To accelerate the convergence of gradient ascent, we represent the primary false peak position in the polar coordinate domain, as the objective function varies more rapidly in the angular domain than in the range domain, following the method proposed in \cite{Teng}.

\subsection{Extension for Multiple False Peaks}
In Section III, we derived a tight lower bound $\text{MSE}_{\text{L}}(\mathbf{p}_{\mathrm{U}}, \mathbf{d})$, considering the impact of the primary false peak on MSE. It is known that practical multi-measurement systems often exhibit multiple high false peaks of comparable magnitude. Each of these peaks poses a non-negligible false detection risk, especially in noisy environments. To design a robust optimization algorithm, it is crucial to account for the other high false peaks.

Motivated by this observation, we extend the $\text{MSE}_{\text{L}}(\mathbf{p}_{\mathrm{U}}, \mathbf{d})$ to incorporate a set of $W$ high false peaks. Based on the definition in \eqref{eq:false_peaks}, we define the false detection event for the $w$-th false peak $\mathbf{p}_{\mathrm{F},w}$ as
\begin{align}
\mathcal{E}_{\mathrm{F},w} \triangleq \left\{\mathcal{L}(\mathbf{p}_{\mathrm{U}}; \mathbf{d}) \leq \mathcal{L}(\mathbf{p}_{\mathrm{F},w};\mathbf{d})\right\}.
\end{align}
Denote by $\Pr(\mathcal{E}_{\mathrm{F},w}\mid\mathbf{p}_{\mathrm{U}},\mathbf{d})$ the false detection probability of $\mathcal{E}_{\mathrm{F},w}$. We further define the expected squared estimation error when the ML estimator mistakenly selects the $w$-th false peak as
\begin{align}
\text{MSE}_{\mathrm{F},w}(\mathbf{p}_{\mathrm{U}}, \mathbf{d}) \triangleq \mathbb{E}\left[\|\hat{\mathbf{p}}_{\mathrm{ML}} - \mathbf{p}_{\mathrm{U}}\|^2 \mid \mathcal{E}_{\mathrm{F},w}, \mathbf{p}_{\mathrm{U}}, \mathbf{d}\right]. 
\label{eq:MSEF,w_definition}
\end{align}
With the above definitions, an objective function for robust array zooming algorithm is formulated as
\begin{equation}
\begin{split}
   \text{MSE}(\mathbf{p}_{\mathrm{U}}, \mathbf{d}) &= \left(1 - \sum_{w=1}^W \Pr(\mathcal{E}_{\mathrm{F},w}\mid\mathbf{p}_{\mathrm{U}},\mathbf{d})\right)  \text{CRB}(\mathbf{p}_{\mathrm{U}}, \mathbf{d}) \\
   &+ \sum_{w=1}^W \Pr(\mathcal{E}_{\mathrm{F},w}\mid\mathbf{p}_{\mathrm{U}},\mathbf{d}) \, \text{MSE}_{{\mathrm{F}},w}(\mathbf{p}_{\mathrm{U}}, \mathbf{d}).
\label{eq:weighted_mse_sum} 
\end{split}
\end{equation}
Based on \eqref{eq:weighted_mse_sum}, the optimization problem is reformulated as
\begin{subequations}
\begin{align}
    (\text{P1})  \qquad &\min_{\mathbf{d}} \max_{\mathbf{p} \in \mathcal{R}} \text{MSE}(\mathbf{p}, \mathbf{d}) \\
    &\text{s.t.} \quad \mathbf{d} \in \mathcal{D}^T. 
\end{align}
\end{subequations}
Since \eqref{eq:weighted_mse_sum} requires evaluating the risks from $W$ high false peaks, we extend the primary false peak search algorithm in Section IV B. Specifically, instead of selecting the global maximum from the grid points, we identify the $W$ local maxima to obtain other prominent false peaks.
After generating $\mathcal{G}_w$ by \eqref{eq:centroid}-\eqref{eq:grid_res},  we compute $\mathcal{L}(\mathbf{p}_w; \mathbf{d})$ for grid points $\mathbf{p}_w \in \mathcal{G}_w$ and select $W$ grid points with the highest log-likelihood values:
\begin{equation}
\mathbf{p}_{\mathrm{F},w,\mathrm{ini}} = \arg\max_{\mathbf{p}_w\in\bigcup\mathcal{G}_w} \mathcal{L}(\mathbf{p}_w; \mathbf{d}), \quad w=1,\ldots,W
\label{eq:top_K_selection}
\end{equation}
where $\mathbf{p}_{\mathrm{F},w,\mathrm{ini}}$ serves as the initial point for the GA method to search for the $w$-th false peak position $\mathbf{p}_{{\mathrm{F}},w}$. Then, we obtain $W$ high false peaks $\{\mathbf{p}_{{\mathrm{F}},w}\}_{w=1}^W$.

\begin{algorithm}
    \caption{High False Peaks Search for Multi-measurement System}
\label{alg:peaks_opt}
\begin{algorithmic}[1]
\Require User sample piont $\mathbf{p}_{\mathrm{U}}^{(n)}\in \mathcal{P}_{\mathrm{U}}$, spacing configuration $\mathbf{d}$, $N_{\mathrm{B}, x}$, $ N_{\mathrm{B}, y}$, SNR, $\kappa$, $\epsilon$.
\Ensure $W$ high false peaks $\{\mathbf{p}_{\mathrm{F},w}\}_{w=1}^W$.
        \State Generate $T$ sets $\mathcal{C}^{(t)}$ by \eqref{eq:uv_t_candidates}- \eqref{eq:deviation constraint}.
        \State Calculate tolerances $\varepsilon_u(t,s)$ and $\varepsilon_v(t,s)$ by \eqref{eq:axis_tol}-\eqref{eq:3db}.
        \State Find $W$ false peak
sets $\mathcal{H}_w$ satisfying \eqref{eq:cons}.
        \For{each $\mathcal{H}_w$}
            \State Generate fine grid $\mathcal{G}_w$ around $(\bar{u}_w, \bar{v}_w)$ by \eqref{eq:centroid}-\eqref{eq:grid_res}.
            \State Obtain $\mathbf{p}_{\mathrm{F},w,\mathrm{ini}}$ by \eqref{eq:top_K_selection}.
            \State Obtain the $w$-th false peak $\mathbf{p}_{{\mathrm{F}},w}$ by GA method.
        \EndFor
\end{algorithmic}
\end{algorithm}

\begin{algorithm}
\caption{Proposed Optimization for Solving Problem (P1)}
\label{alg:spacing_opt}
\begin{algorithmic}[2]
\Require User sample points set $\mathcal{P}_{\mathrm{U}}$, $N_{\mathrm{B}, x}$, $ N_{\mathrm{B}, y}$, SNR, $\mathcal{D}$, $T$, $\kappa$, $\epsilon$.
\Ensure Optimal spacing vector $\mathbf{d}_{\text{opt}}$.
\For{$\mathbf{d} \in \mathcal{D}^T$}
    \For{each user sample position $\mathbf{p}_{\mathrm{U}}^{(n)} \in \mathcal{P}_{\mathrm{U}}$}
        \State obtain $\{\mathbf{p}_{{\mathrm{F}},w}\}_{w=1}^W$ by Algorithm 1.
        \For{each $\mathbf{p}_{\mathrm{F},w}$}
            \State Calculate $\Pr(\mathcal{E}_{\mathrm{F},w}\mid\mathbf{p}_{\mathrm{U}}^{(n)},\mathbf{d})$ by \eqref{eq:pf_qfunction_approximation} and
            \Statex \qquad \qquad \quad $\text{MSE}_{\mathrm{F},w}(\mathbf{p}_{\mathrm{U}}^{(n)}, \mathbf{d})$ by \eqref{eq:MSEF,w_definition}.
        \EndFor   
    \EndFor
    \State Calculate $\text{MSE}(\mathbf{p}_{\mathrm{U}}^{(n)}, \mathbf{d})$ by \eqref{eq:weighted_mse_sum}.
\EndFor
\State Return $\mathbf{d}_{\text{opt}} = \arg\min_{\mathbf{d} \in \mathcal{D}^T} \max_{\mathbf{p} \in \mathcal{R}} \text{MSE}(\mathbf{p}, \mathbf{d})$.
\end{algorithmic}
\end{algorithm}

\subsection{Overall Algorithm}
To solve the optimization problem (P1), the user region $\mathcal{R}$ is partitioned into a uniform grid with sampling points $\mathcal{P}_{\mathrm{U}}=\{\mathbf{p}_{\mathrm{U}}^{(1)}, \ldots, \mathbf{p}_{\mathrm{U}}^{(N)}\}$. Considering that the log-likelihood function exhibits higher sensitivity to angular variations than to the range dimension, we represent the user position $\mathbf{p}_{\mathrm{U}}$ in the polar domain as $(\theta_{\mathrm{U}}, \phi_{\mathrm{U}}, r_{\mathrm{U}})$ and construct a 3D discrete grid centered at the user position $(\theta_{\mathrm{U}}, \phi_{\mathrm{U}}, r_{\mathrm{U}})$ as:
\begin{equation} \label{eq:grid_set}
\begin{aligned}
\mathcal{G}_{\mathrm{U}} \triangleq \Big\{ (\theta^q, \phi^k, r^m) \mid \theta^q &= \theta_{\mathrm{U}} + q \Delta \theta, q \in \mathcal{I}_{N_\theta}, \\
\phi^k &= \phi_{\mathrm{U}} + k \Delta \phi, k \in \mathcal{I}_{N_\phi}, \\
r^m &= r_{\mathrm{U}} + m \Delta r, m \in \mathcal{I}_{N_r} \Big\},
\end{aligned}
\end{equation}
where $\mathcal{I}_N = \{ -\frac{N-1}{2}, \dots, \frac{N-1}{2} \}$, $N_\theta, N_\phi$, and $N_r$ denote the number of grid points for the elevation angle, azimuth angle, and radius, respectively. $\Delta \theta$, $\Delta \phi$, and $\Delta r$ denote the respective grid steps.
Each sampling point $\mathbf{p}_{\mathrm{U}}^{(n)} \in \mathcal{P}_{\mathrm{U}}$ ($n = 1, \dots, N$) corresponds to a coordinate triple $(\theta^q, \phi^k, r^m) \in \mathcal{G}_{\mathrm{U}}$ and is expressed as
\begin{equation} \label{eq:coordinate_mapping}
    \mathbf{p}_{\mathrm{U}}^{(n)} = [r^m \sin\theta^q \cos\phi^k, \; r^m \sin\theta^q \sin\phi^k, \; r^m \cos\theta^q]^{\mathrm{T}},
\end{equation}
yielding $N = N_\theta N_\phi N_r$ grid points.

Since the objective function in (P1) depends on the false peak positions, Algorithm~1 is used to identify the $W$ highest false peaks $\{\mathbf{p}_{{\mathrm{F}},w}\}_{w=1}^W$ for any given user sample position $\mathbf{p}_{\mathrm{U}}^{(n)}\in \mathcal{P}_{\mathrm{U}}$ and configuration $\mathbf{d}\in \mathcal{D}^T$. 
These peak positions are used to calculate the corresponding false detection probabilities and estimation errors, which are substituted into \eqref{eq:weighted_mse_sum} to obtain the MSE.  Algorithm~2 summarizes this overall process to determine the optimal spacing $\mathbf{d}_{\text{opt}}$ that minimizes the maximum MSE across all sampling points in $\mathcal{P}_{\mathrm{U}}$. 

The implementation details of the two algorithms are as follows. In Algorithm~1, for each user sample position $\mathbf{p}_{\mathrm{U}}^{(n)}$ and the spacing configuration $\mathbf{d}$, executing line 1 generates $T$ false peak candidate sets $\mathcal{C}^{(t)}$. Lines 2-3 identify $W$ false peak sets $\mathcal{H}_w$. For each set $\mathcal{H}_w$, lines 5 generates a fine grid around the centroid of the set and lines 6-7 select the top grid point and apply the GA method to refine the position, thereby yielding the $w$-th false peak $\mathbf{p}_{{\mathrm{F}},w}$. After processing all $W$ false peak sets, the Algorithm~1 outputs $W$ high false peaks $\{\mathbf{p}_{\mathrm{F},w}\}_{w=1}^W$.
The Algorithm~2 iterates through all spacing candidates and user sampling positions to evaluate the objective function in \eqref{eq:weighted_mse_sum} and determines the optimal solution according to (P1). 

The complexity of Algorithm~1 primarily arises from two stages: 
(i) \textit{Fine grid search}: Evaluating the log-likelihood function in \eqref{eq:loglikelihood_function} over $W$ local fine grid $\mathcal{G}_w$, each with $N_w$ grid points, across $T$ measurements, yielding $\mathcal{O}(W N_w T N_B)$. 
(ii) \textit{Peak refinement}: Applying GA method for $W$ highest grid points, yielding $\mathcal{O}(W N_1 N_2  N_B)$, where $N_1$ and $N_2$ denote the average number of line search steps in the Armijo backtracking and GA iterations, respectively. 
The total complexity of Algorithm 1 is $\mathcal{O}\left( N_BW (T N_w +  N_1 N_2)\right)$.

Algorithm~2 iterates over $\mathbf{d}\in \mathcal{D}^T$ and $N$ user sample positions, where the spacing candidate set $\mathcal{D}$ has size $N_d$. The outer loops result in $\mathcal{O}(N_d^T N)$ iterations. For each iteration, Algorithm 1 is invoked. Therefore, the overall computational complexity is $\mathcal{O}\left( N_d^T N N_BW (TN_w + N_1 N_2) \right)$.

\section{Numerical Results}
In this section, we present numerical results to evaluate the performance of the proposed array zooming optimization for near-field localization with MA arrays. 
\subsection{Simulation Setup and Benchmark Schemes}
The simulation setup consists of a $N_{\mathrm{B}, x} \times N_{\mathrm{B}, y}$ movable antenna array operating at 6 GHz frequency. The UPA of the BS has equal size in both the $x$-axis and the $y$-axis. The user is located within a conical region with a 120-degree apex angle, and the BS-UE distance ranges from 5 meters to 10 meters, ensuring near-field operation conditions. 
To evaluate the MSE via the sampling grid $\mathcal{G}_{\mathrm{U}}$ defined in \eqref{eq:grid_set}, the numbers of grid points are set to $N_\theta = 5$, $N_\phi = 5$, and $N_r = 3$, yielding a total of $N = 75$ sample points.
The respective grid steps are set to $\Delta \theta = 6^\circ$, $\Delta \phi = 18^\circ$, and $\Delta r = 1$ m.
We adopt $T=2$ measurements to enable multi-measurement fusion for high angular and spatial resolution while suppressing false peak. The antenna spacings are selected from $\mathcal{D}=[1\lambda, 1.1\lambda,\ldots,10\lambda]$ with a step size of $0.1\lambda$, and optimized by the proposed algorithm.

\begin{figure}
    \centering
    \captionsetup{labelsep=period, font=small}
    \includegraphics[width=1\linewidth]{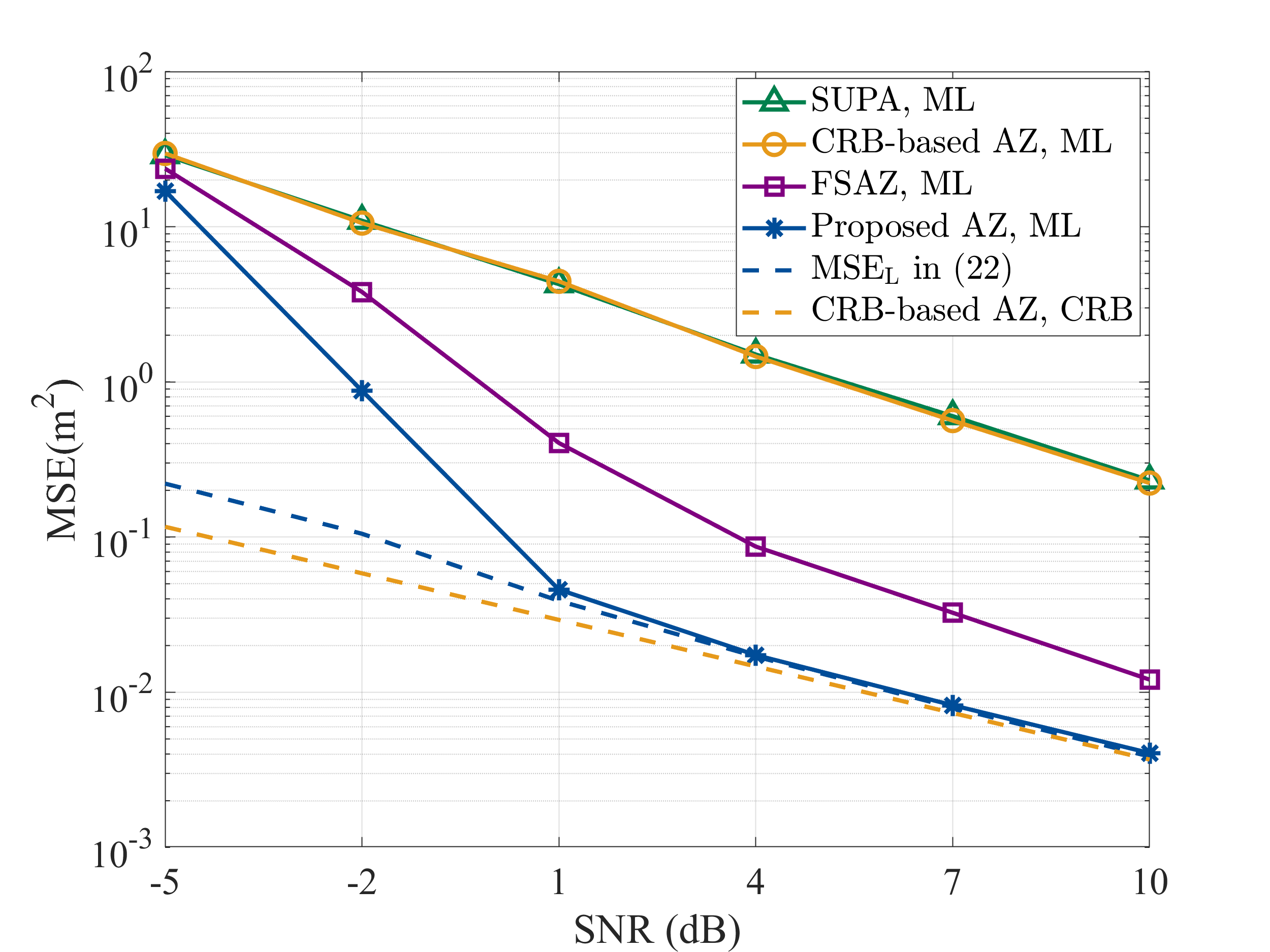}
    \caption{The UE localization performance v.s. SNR.}
    \label{fig:performance}
\end{figure}

\begin{figure}
    \centering
    \captionsetup{labelsep=period, font=small}
    \includegraphics[width=1\linewidth]{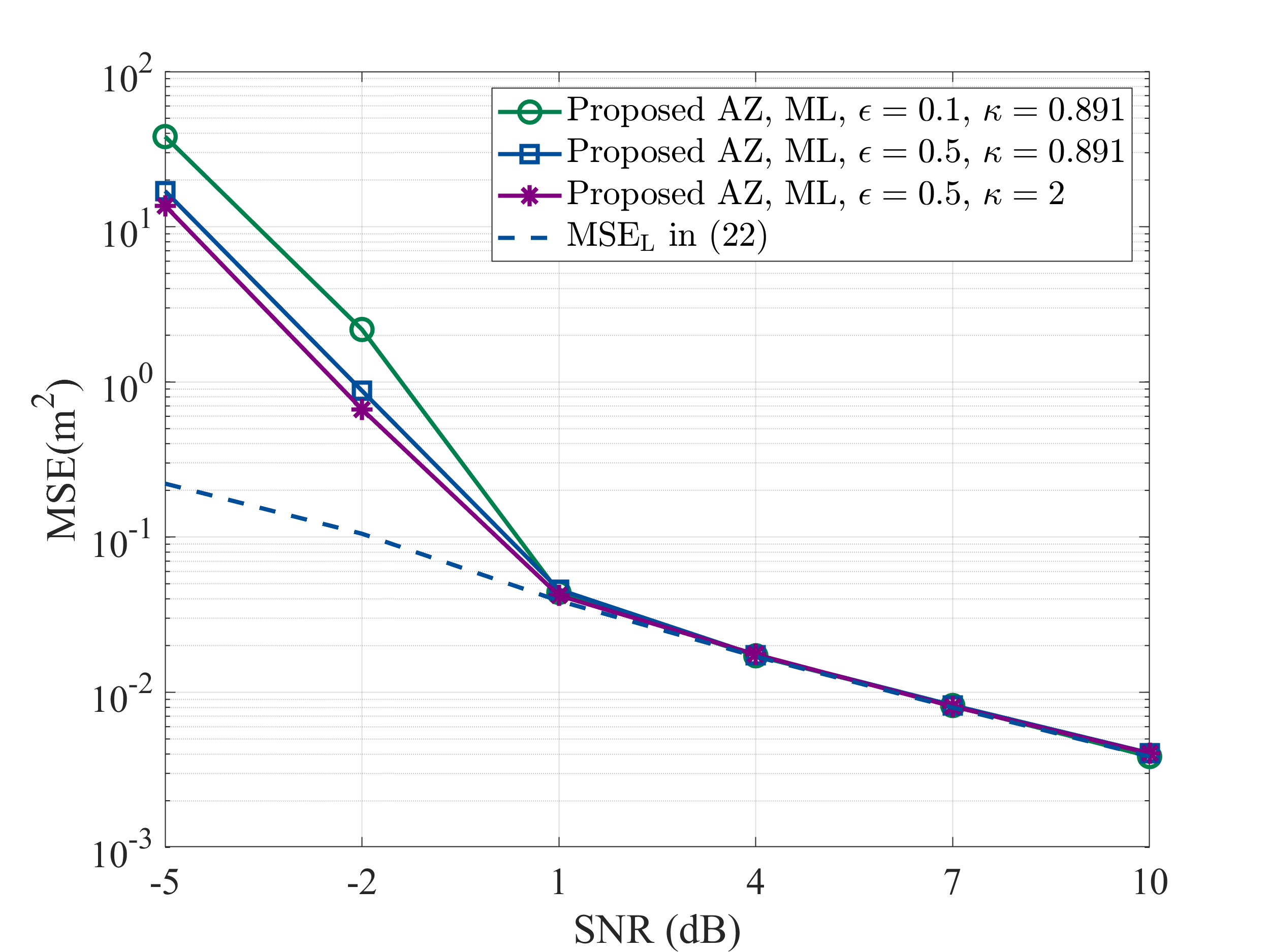}
    \caption{The UE localization performance v.s. SNR with varying $\epsilon$ and $\kappa$.}
    \label{fig:eqsilon_comp}
\end{figure}

To validate the proposed strategy, we compare it with the following baselines:
\begin{enumerate}
    \item Sparse UPA with full aperture (SUPA): The array uses a fixed antenna spacing of $10\lambda$ in each dimension, and only a single measurement is made.
    \item Fixed-spacing array zooming (FSAZ): The scheme employs $T=2$ measurements with a predefined spacing vector $\mathbf{d}=[10\lambda, 1\lambda]$.
    \item CRB-based array zooming: The scheme also adopts $T=2$ measurements and optimizes the spacing vector with the same procedure as the proposed algorithm, except that the CRB is used as the objective function.
\end{enumerate}

For fair comparison, all baselines adopt the same number of antennas as the MA array and all schemes use the ML principle to estimate the user position.

\begin{figure}
    \centering
    \captionsetup{labelsep=period, font=small}
    \includegraphics[width=1\linewidth]{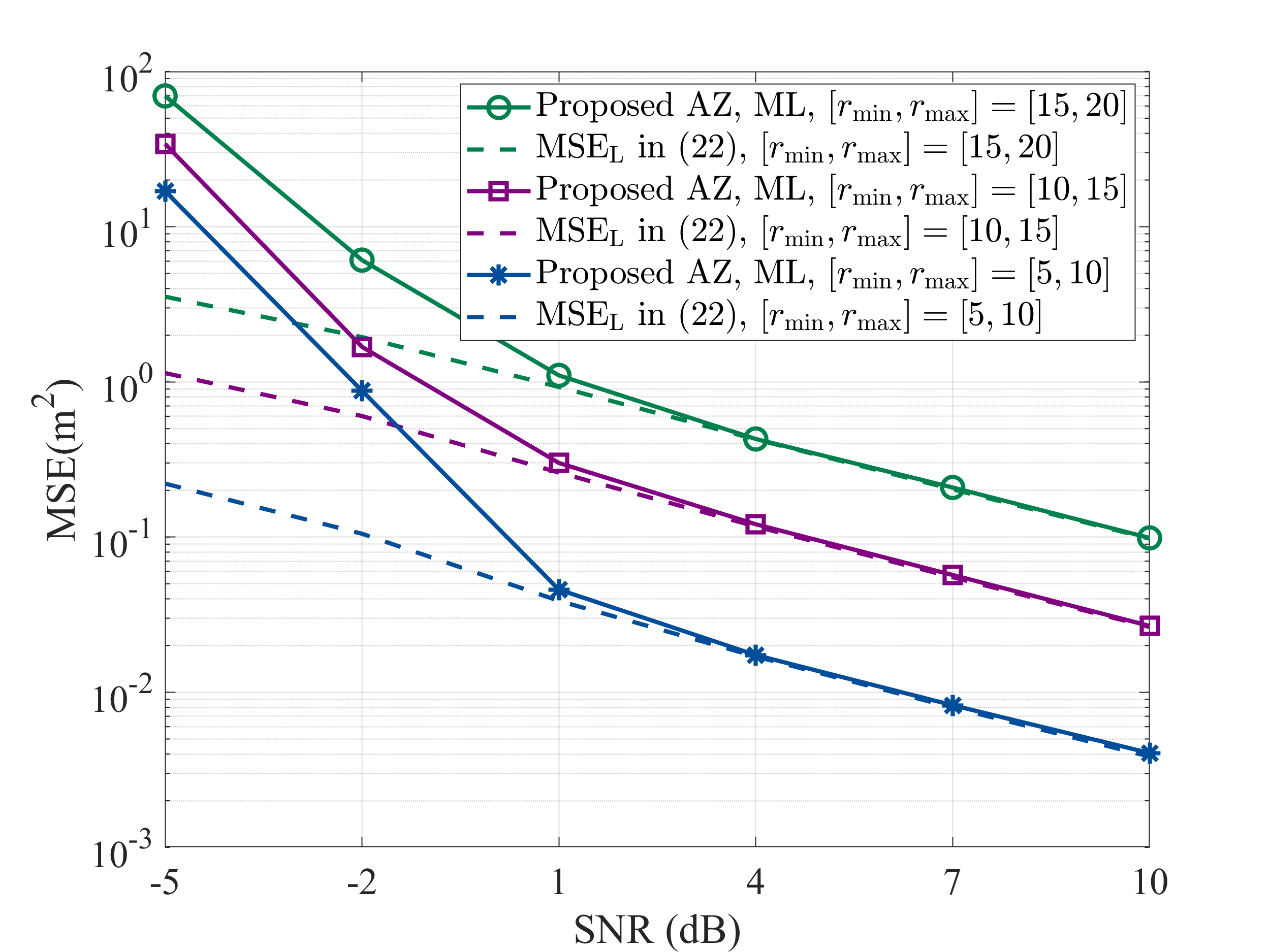}
    \caption{The UE localization performance v.s. SNR with varying the distance range of UE.}
    \label{fig:distance range}
\end{figure}

\begin{figure}
    \centering
    \captionsetup{labelsep=period, font=small}
    \includegraphics[width=1\linewidth]{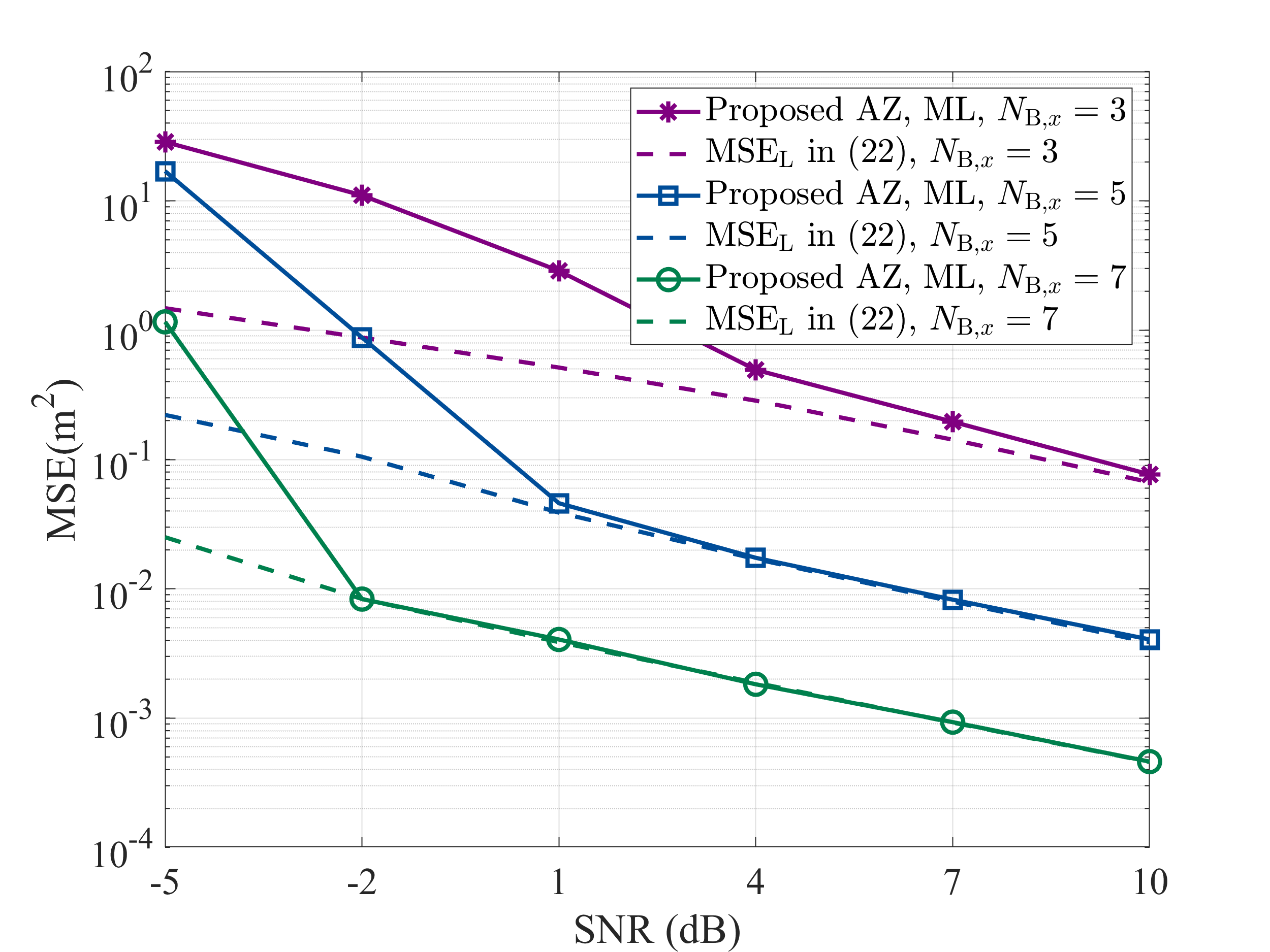}
    \caption{The UE localization performance v.s. SNR with varying the number of MA.}
    \label{fig:number of MA}
\end{figure}

\begin{table}[t]
    \centering
    \caption{Analysis of the optimal spacing $(d^{(2)})^*$ against SNR}
    \label{tab:optimal_spacing_stats}
    \renewcommand{\arraystretch}{1.5}
    \label{tab:optimal_spacing_stats_transposed}
    \setlength{\tabcolsep}{4pt} 
    \begin{tabular}{lcccccc}
    \toprule
    \textbf{Statistical Metric} & \textbf{$-5$ dB} & \textbf{$-2$ dB} & \textbf{$1$ dB} & \textbf{$4$ dB} & \textbf{$7$ dB} & \textbf{$10$ dB} \\
    \midrule
    Mean ($\lambda$) & 8.866 & 8.819 & 8.963 & 9.304 & 9.583 & 9.742 \\
    \midrule
    Std ($\lambda$)  & 1.806 & 1.789 & 0.383 & 0.350 & 0.305 & 0.241 \\
    \bottomrule
    \end{tabular}
\end{table}

\begin{table}[t]
\centering
\caption{False detection probabilities (\%) for different schemes against SNR}
\label{tab:P_f}
\setlength{\tabcolsep}{4pt} 
\begin{tabular}{lcccccc}
\toprule
\textbf{Scheme} & \textbf{$-5$ dB} & \textbf{$-2$ dB} & \textbf{$1$ dB} & \textbf{$4$ dB} & \textbf{$7$ dB} & \textbf{$10$ dB} \\
\midrule
Proposed AZ & 69.18 & 23.00 & 3.28 & 0.00 & 0.00 & 0.00 \\
\midrule
FSAZ & 70.70 & 23.06 & 7.10 & 2.00 & 0.30 & 0.01 \\
\midrule
SUPA & 83.90 & 55.59 & 36.48 & 25.24 & 18.68 & 15.11 \\
\bottomrule
\end{tabular}
\label{table_pf}
\end{table}

\subsection{Numerical Results and Discussions}
\subsubsection{Comparation with baseline schemes}
Fig.~\ref{fig:performance} compares the localization accuracy versus SNR for various baseline schemes and the proposed algorithm. The results show that the optimized scheme outperforms all baseline schemes and approaches $\text{MSE}_\text{L}$ when SNR $\ge0$ dB. At low SNR, more false peaks contribute to MSE, leading to degraded localization accuracy.
The SUPA performs poorly across all SNRs, failing to suppress false peak.
The performance of the FSAZ scheme improves as the SNR increases. The localization performance of the CRB-based array zooming scheme is almost identical to that of SUPA. This is because the CRB only accounts for the estimation resolution, leading to optimal spacings obtained by the optimization being equal to $\mathbf{d}=[10\lambda,10\lambda]$, which is equivalent to SUPA.

The effectiveness of the optimized array zooming strategy is further validated by the false detection probabilities estimated from 1000 Monte Carlo simulations, as illustrated in Table~\ref{tab:P_f}. These values serve as approximations of $\Pr(\mathcal{E}_{\mathrm{F},w}\,|\,\mathbf{p}_{\mathrm{U}},\mathbf{d})$. 
The proposed scheme achieves the lowest false detection probabilities across all SNRs and consequently yields the best localization performance. The FSAZ scheme exhibits a higher false detection probability, while the SUPA scheme shows the highest. The higher false detection probabilities result in the inferior localization performance observed in Fig.~\ref{fig:performance}.

\subsubsection{Optimal Spacings Characteristics Analysis}
To analyze the proposed scheme, we examine the variation of the optimal spacings with respect to SNR.
Through extensive experiments, it is observed that when $T=2$, one of the optimal spacings $(d^{(1)})^*$ is always $10\lambda$, which contributes to maximizing the array aperture and thus achieving higher parameter estimation resolution. Therefore, we focus on the the optimal antenna spacing 
of the other measurement $(d^{(2)})^*$.
Table~\ref{tab:optimal_spacing_stats} presents the mean and standard deviation of $(d^{(2)})^*$, obtained by analyzing the optimal results from 1000 Monte Carlo simulations. 
It can be observed that the proposed scheme adaptively adjusts $(d^{(2)})^*$ according to the SNR. 
At lower SNR, the mean optimal spacing is smaller, indicating that the algorithm prioritizes smaller spacings to suppress false peaks. 
As SNR increases, the mean optimal spacing rises, approaching the maximum aperture and thus improving estimation resolution. 
Furthermore, the proposed scheme exhibits favorable convergence stability across SNRs. 
In the low-SNR regime, the standard deviation is relatively large, reflecting the randomness of optimal solutions under noise-dominated conditions. When SNR $\ge 1$ dB, the standard deviation rapidly decreases, demonstrating that the algorithm comes to a deterministic optimal solution.

\subsubsection{Impact of tolerance $\epsilon$ and beamwidth factor $\kappa$}
Fig.~\ref{fig:eqsilon_comp} investigates the impact of the tolerance parameter $\epsilon$ on localization performance. The parameter $\epsilon$ directly determines the size $W$ of identified high false peaks. A larger $\epsilon=0.5$ retains nearly all candidates, whereas a tighter $\epsilon=0.1$ aggressively filters out more weaker aliases.
When $\text{SNR}\geq0$ dB, both settings converge to $\text{MSE}_{\text{L}}$, as noise fluctuations are insufficient to push ML estimator to lower-magnitude peaks. When $\text{SNR}<0$ dB, the setting with $\epsilon=0.5$ outperforms  that with $\epsilon=0.1$. This is because intense noise can cause any prominent false peak to exceed the true peak. By incorporating a larger set of false peaks into the objective $\text{MSE}$, the algorithm with $\epsilon=0.5$ mitigates diverse false detection risks more effectively. 
$\kappa=2$ corresponds to the first-null bandwidth \cite{beamwidth}. The beamwidth factor $\kappa$ affects the search range of high false peaks, but exerts only a limited effect on performance improvement since it does not determine the number of search regions $W$.

\subsubsection{User location}
We evaluate the localization performance of the proposed algorithm for the user in different regions $[r_{\min}, r_{\max}]$, as illustrated in Fig.~\ref{fig:distance range}. MSE degrades as the UE distance range $[r_{\min}, r_{\max}]$ increases, which stems from the weakened spherical wave-front curvature. As the user moves toward the far-field boundary, phase variations across the ELAA become increasingly linear, reducing range-related Fisher information and deteriorating localization performance.

\subsubsection{BS array size}
We examine the impact of the MA count, $N_{\mathrm{B},x}$, on the proposed algorithm's performance. Fig.~\ref{fig:number of MA} shows that increasing the number of MAs improves the localization accuracy. As $N_{\mathrm{B},x}$ increases, the algorithm converges to the CRB at lower SNR levels. This shows that scaling the MA count not only improves asymptotic precision but also improves the search capability of the algorithm in noise-intensive regimes. A denser MA array sharpens main-lobe resolution and effectively suppresses false peaks, reducing false detection probability and yielding accuracy gains.

\section{Conclusion}
This paper addressed the resolution-ambiguity dilemma in sparse arrays and high hardware cost of large-scale arrays for near-field localization. We have proposed a movable antenna-based array zooming system. We analyzed near-field false peak distribution, derived a tight lower bound $\text{MSE}_\text{L}$ incorporating the false detection probability, and extended it to $\text{MSE}$ considering multiple false peaks. Based on $\text{MSE}$, we propose an optimization algorithm for the multi-measurement array zooming system to suppress false peaks
and minimize the localization error. 
Numerical results demonstrated that the optimized array zooming system effectively reduces the false detection probability and substantially outperforms fixed-spacing arrays and CRB-based baselines.

\appendices

\section{Proof of proposition 1}
\label{app:(14)}
Substituting \eqref{eq:dis} and \eqref{eq:steering_t} into the objective function of \eqref{eq:grating_lobe_optimization} and noting that at a local maximum the phases of the array response vectors must align up to a constant modulo $2\pi$, we obtain the following condition for any solution $\mathbf{p}$ to problem \eqref{eq:grating_lobe_optimization}:
\begin{equation}
\frac{2\pi}{\lambda} \left(r_{(i,j)}^{(t)}(\mathbf{p}_{\mathrm{U}}) - r_{(i,j)}^{(t)}(\mathbf{p})\right) = 2k_{i,j}\pi + \delta, \quad \forall i,j \in \mathcal{I}_{N}
\label{eq:general_phase_cond}
\end{equation}
where $k_{i,j} \in \mathbb{Z}$ and $\delta$ is a constant.
By substituting \eqref{eq:dis} into \eqref{eq:general_phase_cond}, we can reorganize the expression into a polynomial with respect to the indices $i$ and $j$:
\begin{equation}
C_0 + C_1 i + C_2 j+ C_3 i^2 + C_4 j^2 + C_5 ij= k_{i,j}\lambda, \quad \forall i,j \in \mathcal{I}_{N},
\label{eq:LHS}
\end{equation}
with the coefficients defined as:
\begin{subequations}
\begin{align}
C_0 &= r - r_{\mathrm{U}} - \frac{\lambda}{2\pi}\delta, \\
C_1 &= (u_{\mathrm{U}} - u)d^{(t)}, \\
C_2 &= (v_{\mathrm{U}} - v)d^{(t)},\\
C_3 &= \left(\frac{1-u_{\mathrm{U}}^2}{r_{\mathrm{U}}} - \frac{1-u^2}{r}\right)\frac{(d^{(t)})^2}{2}, \\
C_4 &= \left(\frac{1-v_{\mathrm{U}}^2}{r_{\mathrm{U}}} - \frac{1-v^2}{r}\right)\frac{(d^{(t)})^2}{2}, \\
C_5 &= -\left(\frac{u_{\mathrm{U}} v_{\mathrm{U}}}{r_{\mathrm{U}}} - \frac{uv}{r}\right)(d^{(t)})^2.
\end{align}
\label{eq:coefficients}
\end{subequations}

Since \eqref{eq:LHS} holds for all $i \in \mathcal{I}_{N_{\mathrm{B}, x}}$ and $j \in \mathcal{I}_{N_{\mathrm{B}, y}}$, it specifically holds for the subset of indices $i, j \in \{-1, 0, 1\}$. 
By evaluating the polynomial at these points, we can obtain constraints on the coefficients that must hold for a local maximum.

1) By evaluating at the origin $(i=0, j=0)$, we obtain 
\begin{equation}
    C_0 = k_{0,0}\lambda, \quad k_{0,0}\in \mathbb{Z}
\label{eq:c0}
\end{equation}

2) Substituting \eqref{eq:c0} into evaluating along the $x$-axis at $(i=1, j=0)$ and $(i=-1, j=0)$, we obtain:
\begin{subequations}
\begin{align}
C_1 + C_3 &= (k_{1,0}-k_{0,0})\lambda, \quad k_{1,0}\in \mathbb{Z} \label{eq:c1+c3}\\ 
-C_1 + C_3 &= (k_{-1,0}-k_{0,0})\lambda, \quad k_{-1,0}\in \mathbb{Z}
\end{align}
\label{eq:c1_c3}
\end{subequations} 

3) Substituting \eqref{eq:c0} into evaluating along the $y$-axis at $(i=0, j=1)$ and $(i=0, j=-1)$, we can obtain:
\begin{subequations}
\begin{align}
C_2 + C_4 &= (k_{0,1}-k_{0,0})\lambda,
\quad k_{0,1}\in \mathbb{Z}
\label{eq:c2+c4}\\ 
-C_2 + C_4 &= (k_{0,-1}-k_{0,0})\lambda,
\quad k_{0,-1}\in \mathbb{Z}
\end{align}
\label{eq:c2_c4}
\end{subequations}
\setlength{\parskip}{0pt plus 1pt}
4) Substituting \eqref{eq:c0} into evaluating the diagonal cross-term $(i=1, j=1)$, we have
\begin{equation}
C_1 + C_2 + C_3 + C_4 + C_5 = (k_{1,1}-k_{0,0})\lambda, \quad k_{1,1}\in \mathbb{Z}
\label{eq:c1_c5}
\end{equation}

From \eqref{eq:c1_c3}-\eqref{eq:c1_c5}, we can obtain: 
\begin{subequations}
\begin{align}
2C_1 &= \left((k_{1,0}-k_{0,0})-(k_{-1,0}-k_{0,0})\right)\lambda, \label{eq:c1}\\ 
2C_2 &= \left((k_{0,1}-k_{0,0})-(k_{0,-1}-k_{0,0})\right)\lambda,  \label{eq:c2}\\
2C_3 &= \left((k_{1,0}-k_{0,0})+(k_{-1,0}-k_{0,0})\right)\lambda, \label{eq:c3}\\
2C_4 &= \left((k_{0,1}-k_{0,0})+(k_{0,-1}-k_{0,0})\right)\lambda, \label{eq:c4}\\
C_5 &= \left((k_{1,1} - k_{0,0})-(k_{1,0}-k_{0,0})-(k_{0,1}-k_{0,0})\right)\lambda
\label{eq:c5}
\end{align}
\label{eq:c1_c5_result}
\end{subequations}
From \eqref{eq:c1_c5_result}, it is obvious that $2C_1, 2C_2, 2C_3, 2C_4$, and $C_5$ are integer multiples of $\lambda$. Let 
\begin{align}
     2C_1 = k_1\lambda, 2C_2 = k_2\lambda, 2C_3 = k_3\lambda, 2C_4 = k_4\lambda, C_5 = k_5\lambda,
    \label{eq:c_k}
\end{align}
where $k_1, \dots, k_5 \in \mathbb{Z}$. Substituting \eqref{eq:c_k} and \eqref{eq:c0} into the left-hand side of \eqref{eq:LHS}, the terms associated with $i$ and $j$ can be rewritten as 
\begin{equation}
    k_{0,0}\lambda+\frac{\lambda}{2}(k_1 i + k_3 i^2) + \frac{\lambda}{2}(k_2 j + k_4 j^2) + k_5 \lambda ij. 
    \label{eq:52}
\end{equation}
Notice from \eqref{eq:c1} and \eqref{eq:c3} that $k_3 - k_1 = 2(k_{-1,0} - k_{0,0})$, which is an even integer. This indicates that $k_1$ and $k_3$ share the same parity. Consequently, for any integer $i\in \mathcal{I}_N$, the term $k_1 i + k_3 i^2 = k_1(i+i^2) + (k_3-k_1)i^2$ is always an even integer as $i+i^2 = i(i+1)$ is inherently even. This mathematical property guarantees that $\frac{\lambda}{2}(k_1 i + k_3 i^2)$ remains an integer multiple of $\lambda$ for any arbitrary index $i\in \mathcal{I}_{N_{\mathrm{B}, x}}$. By symmetry, the same result holds for the term $\frac{\lambda}{2}(k_2 j + k_4 j^2)$ for any arbitrary index $j\in \mathcal{I}_{N_{\mathrm{B}, y}}$. Thus evaluating \eqref{eq:LHS} for indices outside the subset $\{-1, 0, 1\}$ does not impose any additional constraints on the coefficients in \eqref{eq:coefficients}.  
By substituting \eqref{eq:coefficients} into \eqref{eq:c_k}, we obtain the conditions presented in \eqref{eq:conditions}.

The derivation above has provided the necessary conditions for $\mathbf{p}$ to be a solution to problem \eqref{eq:grating_lobe_optimization}. To prove sufficiency, we assume $\mathbf{p}$ satisfies \eqref{eq:conditions}. Substituting \eqref{eq:conditions} into \eqref{eq:coefficients}, we equivalently have \eqref{eq:c_k} for integers $k_1, \dots, k_5$. Substituting \eqref{eq:c_k} into the left-hand side of \eqref{eq:LHS} yields \eqref{eq:52}.
As proven earlier, the terms $\frac{\lambda}{2}(k_1 i + k_3 i^2)$ and $\frac{\lambda}{2}(k_2 j + k_4 j^2)$ are always integer multiples of $\lambda$, then \eqref{eq:52} evaluates to an integer multiple of $\lambda$ for all $i \in \mathcal{I}_{N_{\mathrm{B}, x}}$ and $j \in \mathcal{I}_{N_{\mathrm{B}, y}}$. Substituting this result back into the left-hand side of \eqref{eq:general_phase_cond} yields $2\pi k_{i,j} + \delta$, which proves the sufficiency. 

As a result, the set of conditions in \eqref{eq:conditions} is both necessary and sufficient for $\mathbf{p}$ to be a solution to problem \eqref{eq:grating_lobe_optimization}. This completes the proof of Proposition~1.

\section{Proof of theorem 1}
\label{app:pf_derivation}
The false detection probability can be expressed as
\begin{equation}
\Pr(\mathcal{E}_{\mathrm{F},1}\,|\,\mathbf{p}_{\mathrm{U}},\mathbf{d})= \Pr(\mathcal{L}(\mathbf{p}_{\mathrm{F},1};\mathbf{d})>\mathcal{L}(\mathbf{p}_{\mathrm{U}};\mathbf{d})).
\label{eq:23}
\end{equation}
Based on \eqref{eq:loglikelihood_function}, $\Pr(\mathcal{E}_{\mathrm{F},1}\,|\,\mathbf{p}_{\mathrm{U}},\mathbf{d})$ is equal to
\begin{equation}
\Pr\left(\sum_{t=1}^T \frac{|\langle \mathbf{A}^{(t)}_{\mathrm{F},1}, \mathbf{Y}^{(t)} \rangle|^2}{\sigma^2\|\mathbf{A}_{\mathrm{F},1}^{(t)}\|_F^2} > \sum_{t=1}^T \frac{|\langle \mathbf{A}^{(t)}_{\mathrm{U}}, \mathbf{Y}^{(t)} \rangle|^2}{\sigma^2\|\mathbf{A}^{(t)}_{\mathrm{U}}\|_F^2}\right).
\label{eq:P_f inequality}
\end{equation}
We define normalized projections as
\begin{equation}
x_{\mathrm{F},1}^{(t)} = \frac{\langle \mathbf{A}^{(t)}_{\mathrm{F},1}, \mathbf{Y}^{(t)} \rangle}{\sigma \|\mathbf{A}^{(t)}_{\mathrm{F},1}\|_F} \quad \text{and} \quad x_{\mathrm{U}}^{(t)} = \frac{\langle \mathbf{A}^{(t)}_{\mathrm{U}}, \mathbf{Y}^{(t)} \rangle}{\sigma \|\mathbf{A}^{(t)}_{\mathrm{U}}\|_F}. 
\end{equation}
Then, the false detection probability can be written as 
\begin{equation}
\Pr(\mathcal{E}_{\mathrm{F},1}\mid\mathbf{p}_{\mathrm{U}},\mathbf{d})  = \Pr\left( \sum_{t=1}^T \left( |x_{\mathrm{F},1}^{(t)}|^2 - |x_{\mathrm{U}}^{(t)}|^2 \right) > 0 \right).
\end{equation} 
Let $\bm{x}^{(t)} = [x_{\mathrm{F},1}^{(t)}, x_{\mathrm{U}}^{(t)}]^{\mathrm{T}}$. Since $\mathbf{Y}^{(t)}$ is a complex Gaussian random matrix according to  \eqref{eq:signal_t}, the $x_{\mathrm{F},1}^{(t)}$ and $x_{\mathrm{U}}^{(t)}$, being linear combinations of $\mathbf{Y}^{(t)}$, also follow complex Gaussian distributions.  Then, the distribution of $\bm{x}^{(t)}$ is given by
\begin{equation} \bm{x}^{(t)} \sim \mathcal{CN}(\boldsymbol{\mu}^{(t)}, \mathbf{\Sigma}^{(t)}),
\end{equation} 
where $\mathbf{\Sigma}^{(t)}= 
\begin{bmatrix} 1 & \rho^{(t)} \\ (\rho^{(t)})^* & 1 \end{bmatrix}$ with the spatial correlation coefficient $\rho^{(t)}=\frac{\langle \mathbf{A}^{(t)}_{\mathrm{F},1}, \mathbf{A}_{\mathrm{U}}^{(t)} \rangle}{\|\mathbf{A}_{\mathrm{F},1}^{(t)}\|_F \|\mathbf{A}_{\mathrm{U}}^{(t)}\|_F}$.
The mean vector $\boldsymbol{\mu}^{(t)}$ is given by
\begin{equation} 
\bm{\mu}^{(t)} = \mathbb{E}[\bm{x}^{(t)}] = \frac{\beta^{(t)}}{\sigma} \|\mathbf{A}_{\mathrm{U}}^{(t)}\|_F 
\begin{bmatrix} 
\rho^{(t)} \\ 1 
\end{bmatrix} .
\end{equation}
Let $z^{(t)} = |x_{\mathrm{F},1}^{(t)}|^2 - |x_{\mathrm{U}}^{(t)}|^2$, which can be written as $z^{(t)} = (\bm{x}^{(t)})^{\mathrm{H}} \mathbf{J} \bm{x}^{(t)}$ with $ \mathbf{J} = \begin{bmatrix} 
1 & 0 \\ 0 & -1 
\end{bmatrix}$. 
The characteristic function of a complex non-central Gaussian quadratic form is given by \cite{CF}
\begin{equation} 
\phi_{z^{(t)}}(u;\mathbf{p}_{\mathrm{U}},\mathbf{d})  = \frac{\exp\left( \mathrm{j}u (\boldsymbol{\mu}^{(t)})^{\mathrm{H}} \mathbf{J} (\mathbf{I} - \mathrm{j}u \mathbf{\Sigma}^{(t)} \mathbf{J})^{-1} \boldsymbol{\mu}^{(t)} \right)}{\det(\mathbf{I} - \mathrm{j}u \mathbf{\Sigma}^{(t)} \mathbf{J})}.
\end{equation} 

As $T$ measurements are independent, the characteristic function of  $z = \sum_{t=1}^T z^{(t)}$ is given by
\begin{equation} \phi_z(u;\mathbf{p}_{\mathrm{U}},\mathbf{d}) = \prod_{t=1}^T \phi_{z^{(t)}}(u;\mathbf{p}_{\mathrm{U}},\mathbf{d}).
\end{equation} 
Using the Gil-Pelaez inversion formula for the cumulative distribution function \cite{Gil_Pelaez}, the probability is given by
\begin{equation} 
\Pr(\mathcal{E}_{\mathrm{F},1}\mid\mathbf{p}_{\mathrm{U}},\mathbf{d}) = \frac{1}{2} + \frac{1}{\pi} \int_0^\infty \frac{\text{Im}[\phi_z(u;\mathbf{p}_{\mathrm{U}},\mathbf{d})]}{u} du.
\end{equation}
This completes the proof of Theorem~1.

\section{Proof of theorem 2}
\label{app:pf_derivation_Q}
Based on \eqref{eq:P_f inequality}, the false detection probability can be represented as
\begin{equation}
\Pr\left(\sum_{t=1}^T \frac{|\langle \mathbf{A}^{(t)}_{\mathrm{F},1}, \mathbf{Y}^{(t)} \rangle|^2}{\|\mathbf{A}_{\mathrm{F},1}^{(t)}\|_F^2} > \sum_{t=1}^T \frac{|\langle \mathbf{A}^{(t)}_{\mathrm{U}}, \mathbf{Y}^{(t)} \rangle|^2}{\|\mathbf{A}^{(t)}_{\mathrm{U}}\|_F^2}\right).
\label{eq:P_f_inequality_2}
\end{equation}
For the received signal model in \eqref{eq:signal_t}, the inner product magnitude for any candidate $\mathbf{p}$ is identically expanded as
\begin{equation}
\begin{aligned}
&|\langle \mathbf{A}(\mathbf{p}, d^{(t)}), \mathbf{Y}^{(t)} \rangle|^2 = |\beta^{(t)}|^2 |\langle \mathbf{A}(\mathbf{p}, d^{(t)}), \mathbf{A}_{\mathrm{U}}^{(t)} \rangle|^2 \\
&+ 2 \Re\left\{ (\beta^{(t)} \langle \mathbf{A}(\mathbf{p}, d^{(t)}), \mathbf{A}_{\mathrm{U}}^{(t)} \rangle)^* \langle \mathbf{A}(\mathbf{p}, d^{(t)}), \mathbf{N}^{(t)} \rangle \right\} \\
&+ |\langle \mathbf{A}(\mathbf{p}, d^{(t)}), \mathbf{N}^{(t)} \rangle|^2.
\end{aligned}
\label{eq:inner_product_expansion}
\end{equation}
Substituting this expansion into \eqref{eq:P_f inequality}, we obtain
\begin{equation}
\Pr(\mathcal{E}_{\mathrm{F},1}\mid\mathbf{p}_{\mathrm{U}},\mathbf{d}) = \Pr\left(N(\mathbf{p}_{\mathrm{U}}, \mathbf{d}) + \Delta N(\mathbf{p}_{\mathrm{U}}, \mathbf{d}) > S(\mathbf{p}_{\mathrm{U}}, \mathbf{d})\right),
\label{eq:PF_de}
\end{equation}
where 
\begin{equation}
N(\mathbf{p}_{\mathrm{U}}, \mathbf{d}) = 2 \Re \left\{ \sum_{t=1}^{T} (\beta^{(t)})^* \langle \mathbf{V}_t, \mathbf{N}^{(t)} \rangle \right\},
\end{equation}
with $\mathbf{V}_t = \frac{\langle \mathbf{A}^{(t)}_{\mathrm{F},1}, \mathbf{A}^{(t)}_{\mathrm{U}} \rangle}{\| \mathbf{A}^{(t)}_{\mathrm{F},1} \|_F^2} \mathbf{A}^{(t)}_{\mathrm{F},1} - \mathbf{A}^{(t)}_{\mathrm{U}}$, 
\begin{equation}
\Delta N(\mathbf{p}_{\mathrm{U}}, \mathbf{d})=\sum_{t=1}^T \left( \frac{|\langle \mathbf{A}^{(t)}_{\mathrm{F},1}, \mathbf{N}^{(t)} \rangle|^2}{\|\mathbf{A}^{(t)}_{\mathrm{F},1}\|_F^2} - \frac{|\langle \mathbf{A}^{(t)}_{\mathrm{U}}, \mathbf{N}^{(t)} \rangle|^2}{\|\mathbf{A}^{(t)}_{\mathrm{U}}\|_F^2} \right)
\end{equation}
and 
\begin{equation}
\begin{aligned}
S(\mathbf{p}_{\mathrm{U}}, \mathbf{d}) = \sum_{t=1}^{T} | \beta^{(t)} |^2 \left( \| \mathbf{A}^{(t)}_{\mathrm{U}} \|_F^2 - \frac{| \langle \mathbf{A}^{(t)}_{\mathrm{F},1}, \mathbf{A}^{(t)}_{\mathrm{U}} \rangle |^2}{\| \mathbf{A}^{(t)}_{\mathrm{F},1} \|_F^2} \right).
\end{aligned}
\end{equation}
Since the elements of the noise matrix $\mathbf{N}^{(t)}$ are identically distributed complex Gaussian random variables with variance $\sigma^2$, $N(\mathbf{p}_{\mathrm{U}}, \mathbf{d})$ is a zero-mean real Gaussian random variable distributed as $\mathcal{N}\big(0, \sigma_{N(\mathbf{p}_{\mathrm{U}}, \mathbf{d})}^2\big)$ with 
\begin{equation}
\sigma_{N(\mathbf{p}_{\mathrm{U}}, \mathbf{d})}^2 = 2\sigma^2\sum_{t=1}^{T}  | \beta^{(t)} |^2  \| \mathbf{V}_t \|_F^2.
\label{eq:N_var}
\end{equation}
Notice that the squared Frobenius norm of $\mathbf{V}_t$ can be simplified as
\begin{equation}
\| \mathbf{V}_t \|_F^2 = \| \mathbf{A}^{(t)}_{\mathrm{U}} \|_F^2 - \frac{| \langle \mathbf{A}^{(t)}_{\mathrm{F},1}, \mathbf{A}^{(t)}_{\mathrm{U}} \rangle |^2}{\| \mathbf{A}^{(t)}_{\mathrm{F},1} \|_F^2}.
\label{eq:V_t}
\end{equation}
Substituting \eqref{eq:V_t} into \eqref{eq:N_var}, we obtain:
\begin{equation}
\sigma_{N(\mathbf{p}_{\mathrm{U}}, \mathbf{d})}^2 = 2 \sigma^2 S(\mathbf{p}_{\mathrm{U}}, \mathbf{d}).
\label{eq:guanxi}
\end{equation}

Observe that $N(\mathbf{p}_{\mathrm{U}}, \mathbf{d})$ is a zero-mean real Gaussian random variable with variance $2\sigma^2 S(\mathbf{p}_{\mathrm{U}}, \mathbf{d})$. In contrast, the quadratic perturbation $\Delta N(\mathbf{p}_{\mathrm{U}}, \mathbf{d})$ consists of quadratic combinations of the Gaussian noise vector, and its variance scales with $\sigma^4$. 
In the high SNR regime, the threshold $S(\mathbf{p}_{\mathrm{U}}, \mathbf{d}) > 0$ in \eqref{eq:PF_de} remains constant. Since the variance of $N(\mathbf{p}_{\mathrm{U}}, \mathbf{d})$ dominates that of $\Delta N(\mathbf{p}_{\mathrm{U}}, \mathbf{d})$ as $\sigma^2 \to 0$, the tail probability of the sum $N(\mathbf{p}_{\mathrm{U}}, \mathbf{d}) + \Delta N(\mathbf{p}_{\mathrm{U}}, \mathbf{d})$ exceeding the threshold is asymptotically determined by $N(\mathbf{p}_{\mathrm{U}}, \mathbf{d})$ alone. Consequently, we obtain the following asymptotic equivalence:
\begin{align}
    &\Pr\big( N(\mathbf{p}_{\mathrm{U}}, \mathbf{d}) + \Delta N(\mathbf{p}_{\mathrm{U}}, \mathbf{d}) > S(\mathbf{p}_{\mathrm{U}}, \mathbf{d}) \big) \nonumber \\
    &\qquad \sim \Pr\big( N(\mathbf{p}_{\mathrm{U}}, \mathbf{d}) > S(\mathbf{p}_{\mathrm{U}}, \mathbf{d}) \big), 
\end{align}
where the notation $\sim$ indicates that the ratio of the two probabilities approaches $1$ as $\sigma^2 \to 0$, i.e.,
\begin{equation}
    \lim_{\sigma^2 \to 0} \frac{\Pr\big( N(\mathbf{p}_{\mathrm{U}}, \mathbf{d}) + \Delta N(\mathbf{p}_{\mathrm{U}}, \mathbf{d}) > S(\mathbf{p}_{\mathrm{U}}, \mathbf{d}) \big)}{\Pr\big( N(\mathbf{p}_{\mathrm{U}}, \mathbf{d}) > S(\mathbf{p}_{\mathrm{U}}, \mathbf{d}) \big)} = 1.
    \label{eq:lim}
\end{equation}
Since $N(\mathbf{p}_{\mathrm{U}}, \mathbf{d})$ follows a real Gaussian distribution, its tail probability can be exactly evaluated using the standard $Q$-function as
\begin{equation}
    \Pr\big( N(\mathbf{p}_{\mathrm{U}}, \mathbf{d}) > S(\mathbf{p}_{\mathrm{U}}, \mathbf{d}) \big) = Q \left( \frac{S(\mathbf{p}_{\mathrm{U}}, \mathbf{d})}{\sigma_{N(\mathbf{p}_{\mathrm{U}}, \mathbf{d})}} \right).
    \label{eq:Q-fun}
\end{equation}
Substituting \eqref{eq:Q-fun} into \eqref{eq:lim}, together with \eqref{eq:guanxi}, we obtain \eqref{eq:pf_qfunction_approximation}.
This completes the proof of Theorem~2.

\end{CJK}
\end{document}